\documentclass[11pt]{article}
\usepackage{a4wide}
\usepackage{amsmath}
\usepackage{amssymb}
\usepackage{amsthm}
\usepackage{natbib}
\usepackage{color}
\usepackage{bm}
\usepackage{graphics,graphicx}
\usepackage{url}
\makeatletter

\newcommand{\para}{\bigskip\noindent}

\newtheorem{theorem}{Theorem}[section]

\newcommand{\be}{\begin{equation}}
\newcommand{\ee}{\end{equation}}
\newcommand{\TITLE}[1]{{\center{\bf#1}\vskip30pt}}

\newcommand{\AFFILIATION}[1]{{\small\center {\it #1} \\}}
\newcommand{\AUTHORS}[1] {{\small\center{#1}\\}}

\newcommand{\mcS}{\mathcal{S}}
\theoremstyle{plain}
\newtheorem{cor}[theorem]{Corrollary}

\begin{document}

\TITLE{Linked shrinkage to improve estimation of interaction effects in regression models}

\AUTHORS{Mark van de Wiel$^{1}$, Matteo Amestoy$^{1}$, Jeroen Hoogland$^{1}$}

\AFFILIATION{$^1$Dep Epidemiology and Data Science, Amsterdam University medical centers, Amsterdam }

\begin{abstract}
We address a classical problem in statistics: adding two-way interaction terms to a regression model. As the covariate dimension increases quadratically, we develop an estimator that adapts well to this increase, while providing accurate estimates and appropriate inference.
Existing strategies overcome the dimensionality problem by only allowing interactions between relevant main effects. Building on this philosophy, we implement a softer link between the two types of effects using a local shrinkage model. We empirically show that borrowing strength between the amount of shrinkage for main effects and their interactions can strongly improve estimation of the regression coefficients. Moreover, we evaluate the potential of the model for inference, which is notoriously hard for selection strategies. Large-scale cohort data are used to provide realistic illustrations and evaluations. Comparisons with other methods are provided.
The evaluation of variable importance is not trivial in regression models with many interaction terms. Therefore, we derive a new analytical formula for the Shapley value, which enables rapid assessment of individual-specific variable importance scores and their uncertainties. Finally, while not targeting for prediction, we do show that our models can be very competitive to a more advanced machine learner, like random forest, even for fairly large sample sizes. The implementation of our method in \texttt{RStan} is fairly straightforward, allowing for adjustments to specific needs.
\end{abstract}



\section{Introduction}
Adding interactions to a regression model is a classical problem in statistics which may lead to interesting insights on the joint effects of covariates \cite[]{afshartous2011key}. It comes at a price though, as the number of interaction terms $q$ increases quadratically with the number of covariates $p$. While one may argue that in very small $p$ settings the plausibility of each two-way interaction may be considered separately, such a strategy is infeasible or unpractical for larger dimensions. At the other end of the spectrum,
with $p$ large - and hence $q$ very large - compared to sample size $n$, the hierarchical lasso \cite[]{bien2013lasso} and variations thereof provide a computationally efficient sparse solution. The latter, however, focuses on selection, and does not come naturally with parameter inference. This leaves a gap for the middle spectrum, with $p+q$ of a similar order of magnitude as $n$, a setting which is fairly common in many biomedical or epidemiological studies. Our goal is to fill this gap using an interpretable solution that on one hand is able to deal with the large number of parameters, while on the other hand allows for inference. More specifically, our aim is three-fold: 1) accurate estimation of parameters; 2) interpretation of variable importance scores in the context of our model; and 3) inference for those variable importance scores.

To achieve these goals, this study provides two novelties. First, a linked shrinkage model, which links local shrinkage of the interaction effects to that of the main effects. This extends the Bayesian local shrinkage framework \cite[]{gelman2008weakly}. The latter provides flexible, differential shrinkage of small and large effects, which may benefit the accuracy of the parameter estimation in the same spirit as the adaptive lasso and non-negative garrotte do. In addition, we draw upon its good inferential properties \cite[]{Wiel2023think}.
The linked shrinkage model also includes a global shrinkage parameter for the interaction parameters to allow those to be weaker on average than the main effect parameters, thereby providing adaptivity. Second, we deduce a computationally efficient equation for Shapley values \cite[]{aas2021explaining}, which allows quantification and inference for those sample specific variable importance scores. Shapley values are popular in machine learning, and we argue that these scores can also be of great use for regression models with many interaction terms, as the presence of the latter impedes straightforward interpretation of the regression coefficients as variable importance scores \cite[]{afshartous2011key}.


As our problem is a classical one in statistics, a number of solutions already exist. Below we provide a list of reference methods that we compare our method to. Here, the first three do not focus on selection, the others do.
First, simple ordinary least squares (\texttt{OLS}), which does not apply any shrinkage, and may therefore provide unstable estimates for some settings. Second, ridge regression with two tuned penalties \cite[]{wood2011fast}, one for main effects, one for interactions: \texttt{ridge2}. Such global penalties likely improve the predictive abilities of the model, but do not well accommodate strong differences between parameter strengths within each of the two parameter sets.
Third, Bayesian local shrinkage \cite[]{gelman2008weakly} using a local Gaussian prior for each parameter, the standard deviations of which are endowed with a half-Cauchy prior: \texttt{Bayloc}. \cite{gelman2008weakly} argue that appropriate standardization (e.g. -1/2, 1/2 for binaries) ``automatically applies more shrinkage to higher-order interactions''. This is true, as such standardization causes two-way interactions (products) to be on a smaller scale than main effects. This does not, however, link the shrinkage of an interaction to that of its corresponding main effects nor adapt it globally to the data. Nevertheless, this model is an important basis for ours.
Fourth, a two-step approach that only include interactions of significant main effects: \texttt{2step}. While popular in practice, it may render very unstable results, as the inclusion of interactions depends on a hard threshold for the main effects. Fifth,
a lasso regression with only a penalty for the interaction terms: \texttt{lassoint}. This type of global shrinkage may suffer from the same drawbacks as \texttt{ridge2}. And sixth, hierarchical lasso for interactions (\texttt{hlasso}), a state-of-the-art methodology that applies the same reasoning as \texttt{2step}, but formalizes it in one fitting procedure \cite[]{bien2013lasso, lim2015learning, du2021lasso}. It is mostly designed for computational efficiency to handle large $p$. While it has proven its use for variable (and parameter) selection, formal inference is far from straightforward \cite[]{lim2015learning}, requiring strong assumptions on the underlying sparsity or extensive resampling.

We compare our linked shrinkage model, termed \texttt{Bayint}, to those methods as well as to several variations of \texttt{Bayint} which differ in how they encode the linked shrinkage. For this, we use the OLS estimates of a very large data set as a benchmark. We study the results for two outcomes (systolic blood pressure and cholesterol), and a mix of continuous, binary and categorical covariates. In addition, we provide several illustrations to support interpretation of the model and the covariates, including those with Shapley values and their uncertainties. While we do not focus on prediction, we perform a short comparison with random forest. This illustrates that even for large sample sizes \texttt{Bayint} can be very competitive to such a machine learner in terms of out-of-bag predictive performance.
We end by discussing the implementation, scalability and potential extensions.

\section{Approach}
The model, called \texttt{Bayint}, combines ideas behind the hierarchical lasso, which considers interactions of strong main effects to be more important, with those of Bayesian local shrinkage, the hierarchical set-up of which allows a softer link between the interactions and main effects.

\subsection{The linked shrinkage model}
For simplicity, we assume linear response $Y_i, i=1, \ldots, n$, but it can easily be reformulated in a generalized linear model or Cox regression context. For sample $i$, the $j$th covariate is denoted by $x_{ij}, j =1, \ldots, p.$ Then, the proposed model is:
\begin{equation}
\begin{split}
Y_i &=\alpha + \sum_{j=1}^{p} \beta_j x_{ij} + \sum_{j,k: j \neq k} \beta_{jk} x_{ij}x_{ik} + \epsilon_i\\
\alpha &\sim N(0,10^2)\\
\beta_{j} &\sim N(0,\sigma^2\tau^2_{j})\\
\beta_{jk} &\sim N(0,\sigma^2\tau_{j}\tau_k\tau_{int})\\
\epsilon_i &\sim N(0,\sigma^2)\\
\tau_{j} &\sim C^{+}(0,1)\\
\tau_{int} &\sim U(0.01,1)\\
\sigma^2 &\sim IG(1,0.001)
\end{split}\label{bayint}
\end{equation}
Several of the components in \eqref{bayint} are in line with conventional Bayesian modelling, including
the half-Cauchy prior on the (relative) standard deviations $\tau_j$ \cite[]{gelman2008weakly}. We add linked shrinkage to the model by including the product $\tau_j\tau_k$ in the prior of $\beta_{jk}$. This product renders a symmetric handling of strong and weak main effects (corresponding to large $\tau_{j}$ and small $\tau_{j}$, respectively), whereas it is on the same scale as each of the components when they are. In addition, $\tau_{int}, 0.01 \leq \tau_{int} \leq 1$ is a shrinkage parameter shared by all interactions that models the prior believe that interaction parameters might, on average, be weaker than the main effect parameters. The lower bound avoids complete shrinkage to 0 of all interaction effects, as this may be undesirable in a sparse setting. Note that when categorical covariates are present, the summation over $j,k$ in the regression model in \eqref{bayint} is adjusted such that interactions between their levels are excluded.

\subsection{Alternative linked shrinkage models}
We discuss a few variations of \texttt{Bayint} \eqref{bayint} that may be relevant for other settings or foci. First, \texttt{Bay0int}, which does not apply shrinkage to the main main effects (non-informative Gaussian priors). This may be useful when one thinks of our model as (a simplification of) a general quadratic form, for which shrinkage of the main effects towards 0 is not necessarily logical. A potential disadvantage is that one looses the link between shrinkage of the two types of effects. Second, \texttt{Bayintadd}, which replaces $\tau_{j}\tau_k$ by $(\tau_j^2 + \tau_k^2)/2$, which lets the strongest main effect dominate the shrinkage of the interaction. That is, if any of the two main effects is strong, this leads to relatively little shrinkage (large prior variance) of $\beta_{jk}$.

If one is particularly interested in detecting interactions, a third alternative may be attractive: \texttt{Bayint*}, which sets $\tau_{int} = 1$. This model does usually not compete with \texttt{Bayint} in terms of prediction accuracy, as the latter has a global shrinkage parameter $\tau_{int}$ that can adapt to the interactions being weaker (on average) than the main effects for most problems. The downside of including $\tau_{int}$, though, is that relatively strong interactions may be over-shrunken, which is why
\texttt{Bayint*} may be better at detecting those. Comparisons with these alternative models are provided further on.


\section{Results}
Here, we assess model \eqref{bayint} in a broad sense by considering parameter estimation and inference, interval coverage, interpretation and prediction. We focus on a realistic data setting for which we may assume (nearly) true values to be known. When appropriate, performance is compared to that of several competitors. These are discussed in the Introduction; more details on their implementations are provided in the Supplementary Material.

\subsection{Data}
The main data we use throughout the manuscript is obtained from the Helius study \cite[]{helius2017}. We use this data set, because it reflects a fairly standard epidemiological study and contains a mix of binary, continuous and categorical covariates.
We consider both systolic blood pressure (log scale; SBP) and cholesterol as response, and
age, gender, ethnicity (5 levels; coded with 4 dummies), smoking (yes/no), packyears, coffee consumption (yes/no), BMI and 4 simulated standard normal noise variables as covariates, rendering $p=14$ covariates. All two-way interactions are considered, except those between the 4 dummy variables representing the categorical covariate, rendering $q = \binom{14}{2} - \binom{4}{2} = 91 - 6 = 85$ interaction parameters.

The entire data set, referred to as the `master set', consists of $N = 21,570$ samples. Therefore, OLS estimates based on the master set are very precise, and hence safely used as benchmarks. As a verification, we confirm that i) the estimated coefficients of the noise variables are indeed very close to zero; and ii) the coefficients estimated by (adaptive) lasso are very close to the OLS estimates (Suppl. Fig. \ref{olsvslassochol} and \ref{olsvslassosbp}).

Continuous covariates were centered and scaled, that is standardized. On the centering, we follow the advise by \cite{afshartous2011key}, as the centering (largely) removes collinearity between main effects and two-way interactions. Scaling is generally applied in shrinkage settings, and also helps to interpret the coefficients and estimation errors relative to one another. Binary covariates were (contrast) coded as -1, 1, which renders them standardized in the balanced setting. For interpretation, we prefer to use the same coding for all binaries. The categorical variable, ethnicity, was contrast-coded with levels -1,0,1.

\subsection{Parameter estimation}
We evaluate parameter estimation of any $\beta = \beta_j$ or $\beta=\beta_{jk}$ by the root Mean Squared Error (rMSE), defined as $$\text{rMSE} =
\sqrt{\frac{1}{B}\sum_{b=1}^B (\hat{\beta}^{(b)} - \beta)^2},$$ with $\hat{\beta}^{(b)}$ the estimator of $\beta$ for the $b$th training set. We used $B=25$ (nearly non-overlapping) training sets of size $n=1,000$, and set the `true' $\beta$ to the OLS estimate from the large master data set. For Bayesian methods, the posterior mean was used
as a point estimate for $\beta$. Figures \ref{rmsecholsplit} and \ref{rmsesbpsplit} compares the results of \texttt{Bayint} with other methods (see Introduction) for cholesterol and SBP as outcome, respectively. The bold line demarcates the main effects and interactions; the thin lines separate the strong effects from the weaker ones, as defined by significance in the master set ($p<0.01$).

Overall, we observe that \texttt{Bayint} shows good performance. The upper displays show that \texttt{Bayint} and \texttt{Bayloc} perform better than OLS across the entire range. \texttt{Bayint} is competitive to \texttt{Bayloc} for the strongest interactions, and superior for most other parameters. The latter situation reverses for the comparison with \texttt{ridge2}, which is competitive to \texttt{Bayint} for most parameters, but not for the important sub-group of strongest interactions, probably due to over-penalization by the global interaction penalty parameter. The latter conclusion is similar for the comparison with \texttt{lassoint}, although here the differences with \texttt{Bayint} are fairly small for SBP (Fig \ref{rmsesbpsplit}). For the latter, the hierarchical lasso, \texttt{hlasso}, is fairly competitive to \texttt{Bayint} on estimating interactions, but lags behind for estimating main effects. Here, we should note that \texttt{hlasso} is more tailored to covariate screening for large $p$, and less so to estimation. For cholesterol (Fig \ref{rmsecholsplit}) there is an additional  small edge for \texttt{Bayint} for the estimation of strong interactions. Finally, \texttt{2step} performs inferior to \texttt{Bayint} for a substantial number of main effects and interactions, and not superior for the others.

Supplementary Figure \ref{mainintconnect} connects the \emph{true} main effects and interactions (estimated from the master set), to provide insight on why linked shrinkage has a benefit for both outcomes. Indeed, we observe that strong interactions tend to link relatively frequently to strong main effects, and that this tendency is somewhat stronger for cholesterol, explaining the slightly larger benefit of linked shrinkage for this outcome compared to SBP.

Finally, we provide a short comparison of \texttt{Bayint} with two aforementioned alternatives, \texttt{Bay0int} and \texttt{Bayintadd}, for the cholesterol model only. From Supp. Fig. \ref{rmsebayintcompetcholsplit} we observe that particularly \texttt{Bayint} and \texttt{Bayintadd} are very competitive, with the latter slightly worse for the very non-significant interactions. \texttt{Bay0int} may pick up the strong interactions slightly better, but seems somewhat inferior for main effects and less important interactions.

\subsection{Inference and interpretation}
Here, we discuss several techniques and visualisations to perform inference and interpret results from the \texttt{Bayint} model. We focus on the model that uses cholesterol as outcome. For inference, we limit the comparison to OLS, given its ubiquitous use and well-known inferential properties, which are less well established for most of the other methods.

\subsubsection{Parameters}
We first briefly discuss inference on the parameters, before extending it to more general variable importance metrics.
Our model renders credible intervals that may be used for this purpose. We previously showed the coverage of Bayesian local shrinkage - on which our shrinkage model is based - to be rather good \cite[]{Wiel2023think} in low dimensional settings, although this will depend on the $p : n$ ratio and the amount of collinearity. Here, we focus on detection, but refrain from a formal comparison with OLS or other frequentist methods (in terms of fixing type I error), given the different perspectives on (multiple) testing.

Supplementary Figure \ref{detections} plots the detections by OLS (criteria: $p\leq 0.05, p \leq 0.01$) against those by \texttt{Bayint} (criterion: 95\% credible interval contains 0) and \texttt{Bayint*} (which fixes $\tau^2_{int} = 1$ in \eqref{bayint}; same criterion). As expected, OLS produces many detections at $p<0.05$, also for effects that are (very) small in the master set (top figure). For $p<0.01$, OLS seems to align better with the results of \texttt{Bayint} (middle figure), in particular for the weak effects and strong interactions, and less so for two strong main effects that are more frequently detected by \texttt{Bayint}. The bottom figure shows that \texttt{Bayint*} indeed detects strong interactions more often than \texttt{Bayint}, at the cost of detecting two main effects (which are involved in those interactions).

Additionally, to be less dependent on the choice of the cut-off, we plot the sensitivities of the methods for specificities ranging from 90\% to 98\%. For that, we define positives as those significant in the master set at cut-off $p \leq 0.05/99$ (Bonferroni correction) and negatives as those that either correspond to a noise covariate or are non-significant at cut-off $0.05$. Given the sheer size of the master set the latter assures that such effects are either very small or completely absent. This defines 12 positives, and 79 negatives; the remaining 99 - 12 - 79 = 8 effects are indeterminate, which are therefore not used for calculating the sensitivities and specificities.
Supplementary Figure \ref{roc} shows these, averaged over 500 subsets, where the curves are parameterized by the thresholds used for detection, either on $p$-values (OLS) or on coverage percentage of the credible interval (\texttt{Bayint} and \texttt{Bayint*}). It clearly shows that for this data set the latter two are better able to separate the true effects from the false ones, given the higher sensitivities across the specificity range.

\subsubsection{Variable importance}
The credible intervals are an important tool to assess the relevance of interaction terms.
Interpretation and inference for the main effect parameters is hampered though by the presence of those interactions, as the effect of one unit change of a covariate depends on the values of the other covariates. Therefore, technically, $\beta_j=0$ only means that for a (fictive) person with average values for all other covariates (given centering is applied), covariate $j$ has no effect. That is, it only quantifies a \emph{conditional} main effect. \cite{afshartous2011key} propose several useful alternatives, such as determining the `range of significance'. For this, one plots the confidence/credible intervals for $\beta_j + \beta_{jk}x_{ik}$ - the effect of one unit change of $x_{ij}$ when interacting with one covariate $x_{ik}$ - against $x_{ik}$. Alternatively, one may compute $E_{ij} = \beta_j + \sum_{k \neq j}\beta_{jk}x_{ik}$, i.e. a `personalized unit change effect' which accounts for all interactions. Our MCMC samples easily provide the posteriors of $E_{ij}$, allowing to plot its uncertainty as well.
A hybrid of the latter two solutions is a plot of $E_{ij}$ against $x_{ik}$ (when continuous) or for color-coded levels of $x_{ik}$ to see whether one unit change of $x_{ij}$ (say age or BMI) has a different effect on the outcome, e.g. for $x_{ik} = -1, 1$ (say female/male), while accounting for the other interacting covariates as well. Supplementary Figure \ref{deltaage} plots $E_{ij}$ and its uncertainty for 100 random test individuals (\texttt{Bayint} model fitted on 1,000 training samples), with $x_{ij}$ and $x_{ik}$ representing age and gender, respectively. We clearly observe a different effect of age increase between genders, but also within gender due to interactions of age with other covariates.

\para
Alternatively, Shapley values \cite[]{aas2021explaining} may be considered. A Shapley value $\phi_{ij}$ quantifies the average contribution of the $j$th feature to the prediction of the $i$th sample, fixing $x_{ij}=x^*_{ij}$. Here, `average' refers to a weighted average over subsets $\mathcal{S}$ that contains all other covariates $(x_{ik})_{k \in \mcS}$  that actively impact the prediction by fixing $x_{ik}=x^*_{ik}$ (called the `players'). Predictions are marginalized over the complement, $\mathcal{S}'$, which defines the non-players $(x_{i\ell})_{\ell \in \mcS'}$, which are considered random. Here, the weights are chosen such that
different sizes of $\mathcal{S}$ have an equal impact on the Shapley value. A formal definition is given in the Supplementary Material. Shapley values are popular in machine learning nowadays, because they uniquely possess several nice properties: efficiency, symmetry, dummy player and linearity \cite[]{aas2021explaining}. Obtaining its exact value is usually computationally very demanding, let alone computing uncertainties. For our model, however, it is feasible to compute Shapley values and their uncertainties efficiently, if one is willing to use the common convention that the marginalization ignores the dependency between the players and the non-players \cite[]{lundberg2017unified}, an approach referred to as `interventional Shapley value' \cite[]{aas2021explaining}. For a linear regression model with two-way interactions and centered covariates it equals
\begin{equation}\label{phipmain}
\phi_{ij} = \beta_j x_{ij}^* + \frac{1}{2}\biggl(\sum_{k: k \neq j}\beta_{jk} x_{ij}^*x_{ik}^* - \sum_{k: k \neq j}\beta_{jk} E[x_{ij}x_{ik}]\biggr),
 \end{equation}
when the $j$th covariate is continuous or binary.
A proof is provided in the Supplementary Material, which also includes expressions for the non-centered setting and categorical covariates. Note that (the posterior of) $\phi_{ij}$ is straightforward to compute after estimating $E[x_{ij}x_{ik}]$ by the sample covariance.
Again, we illustrate results for 100 random test samples and the \texttt{Bayint} model trained on 1,000 random training samples. Figure \ref{shapage} shows Shapley values and their credible intervals for age and Noise.1. The latter is a useful negative control as we observe that, as desired, all credible intervals cover 0. Note that centering of the covariates implies that Shapley values are expected to center around 0. Yet, we observe that age is an important covariate for the majority of samples as most intervals do not cover 0. Supplementary Figures \ref{coverchol} and \ref{coversbp} provide empirical evidence that these intervals, as computed from the output of \texttt{Bayint},  provide appropriate coverage for the majority of covariates and individuals.
Clearly, $\phi_{ij} = \phi_{ij}^{\text{main}} + \phi_{ij}^{\text{int}}$, which denote the contributions of the main effect and that of all interactions with covariate $j$. Supplementary
Figure \ref{shapleyall} displays the Shapley values (posterior means), and its two contributors, for all covariates.
Alternatively, Fig. \ref{shapimportance} shows the conventional variable importance derived from Shapley values, $I_j = 1/n \sum_{i=1}^n |\phi_{ij}|$, and analogously defined, $I^{\text{main}}_j$ and $I^{\text{int}}_j$. Note that, in general, $I_j \neq I^{\text{main}}_j + I^{\text{int}}_j$. Still, plotting both $I^{\text{main}}_j$ and $I^{\text{int}}_j$ renders insight on how relevant the main effects and interactions are for each covariate. While the main effects show the strongest importance scores for most covariates (except BMI), we do observe that interactions are also relevant for a fair share of the covariates.

\subsection{Model assessment by $R^2$}
Although our focus lies on parameter estimation, it is useful to compare the overall fit of \texttt{Bayint} with those of i) a basic regression model without interactions; and ii) with more advanced machine learners, such as the random forest. The first comparison allows to judge the additive value of the interactions for improving test sample fit. The second one is relevant, because the random forest holds the promise to capture interactions well and to provide adequate predictions, so it provides a useful benchmark. As $p=14$ is small relative to $n=1,000$, we simply used OLS for the main effects model; random forest was either fit using the defaults of the \texttt{rfsrc} function in the \texttt{randomForestSRC} package (\texttt{RF}) or with hyperparameters (\texttt{mtry}: number of features considered per split and \texttt{nodesize}: minimum node size) tuned for optimal predictive performance using the \texttt{tune.rfsrc} function  (\texttt{RFtune}).

For the comparison, we compute for each model and for all $b = 1, \ldots, 25$ training sets the out-of-bag coefficient of determination, $R^2_b = 1 - \sum_{i \in \mathcal{T}_b} (y_i - \hat{y}_{i,b})^2)/\sum_i(y_i - \bar{y})^2$, with $\mathcal{T}_b$ the set of all out-of-bag samples for training $b$, and $\hat{y}_{i,b}$ the prediction for test sample $i$ by the $b$th model. Figure \ref{r2s} shows the results. In-bag predictive performance is shown as well, to illustrate potential overfitting.

For cholesterol as outcome, we observe that \texttt{Bayint} provides a substantial gain in terms of $R^2$ as compared to the main effects only model. Out-of-bag predictive performance is somewhat better than that of \texttt{RF}, and marginally better than that of \texttt{RFtune}. The latter two overfit substantially, as observed from comparing the in-bag and out-of-bag performances. For SBP as outcome, differences between \texttt{Bayint} and the main effects model are small, whereas both beat \texttt{RF} and \texttt{RFtune} by a fair margin. Note that the relative performance of the random forest improves for $n=5,000$ (Supplementary Figure \ref{r2s5000}), rendering it competitive to \texttt{Bayint} in terms of prediction.



\section{Implementation and data availability}
Our linked shrinkage model was implemented in \texttt{RStan} (v 2.21.8) \cite[]{RStan}. We chose to use a general purpose sampler for several reasons. First, it allows the user to adjust the model without much extra effort in terms of fitting or inference. This includes variations on modelling the shrinkage, as illustrated, but also adjusting the likelihood to allow binary or survival outcome, as \texttt{RStan} accommodates these as well. We provide an example for logistic regression in the code. Second, \texttt{RStan} provides several diagnostic tools, such as trace plots, to check the convergence of the MCMC sampler. Our scripts are available at \url{https://github.com/markvdwiel/ThinkInteractions/}, which also contains a synthetic version of the data set. The covariates in the synthetic data are generated by imputation as described in \cite{Wiel2023think}. Then, responses (Cholesterol and SBP) are generated by drawing from normal distributions with means equal to the OLS (as fitted on the master set) predictions based on the synthetic covariates, and error variance equal to the residual error variance of the OLS model. We verified that the synthetic set renders qualitatively similar results on the regression models as those presented here (cf. Figure \ref{rmsecholsplit} and Supp. Fig. \ref{rmsecholsplit_synth};
Figure \ref{rmsesbpsplit} and Supp. Fig. \ref{rmsesbpsplit_synth}). Finally, as an indication: running time for our example data sets of $n=1,000, p=14, q=85$ is around 3min for 25,000 MCMC samples using a single core from a PC with a 1.30 GHz processor and 16Gb RAM.

\section{Discussion}
We demonstrated the potential of linked shrinkage for improving parameter estimating in fairly large regression models that include all two-way interactions. Naturally, the benefit depends on the data set and the relevance of those interactions, in particular in connection to the main effects. A limitation of our main model, \texttt{Bayint}, is that is not good in discovering interactions for which none of the main effects are relevant. Therefore, we offer an alternative, \texttt{Bayint0}, which suits such a setting better while still providing a parsimonious model for the shrinkage of interaction effects. Moreover, if one knows about such an interaction, it can be taken apart, in terms of shrinkage. Another limitation might be the linear scale of the covariates in the model. As regression comes with many tools for model diagnostics, such a model miss-specification can be diagnosed. In principle, it is straightforward to extend our model by adding non-linear transformations (log, quadratic) of a few covariates, each with their own local penalty. Optimizing the model by including such covariate transformations after viewing the data may come at the cost of invalid inference. Possibly, for large $n$, the cost of this form of data snooping may be limited when only the main effects are considered, but this requires further study. In addition, note also that if the model is very wrong its predictive performance is likely compromised, as compared to more flexible machine learners. This was not the case for our data, at least in comparison to the random forest.

A natural extension is the inclusion of higher-order interaction. In principle, this is easily achieved with our code, and may be a viable option when only few covariates are present. Note, however, that the number of three-way interactions increases cubicly with the number of covariates, and one may argue that regression models with higher-order interaction are hardly any easier to interpret than many machine learners.

As mentioned, we coded our method in \texttt{RStan}, because its flexibility allows the user to extend the framework to one's own needs, such as the various illustrated shrinkage links and non-continuous outcomes (survival, binary). A potential disadvantage of using such a general purpose sampler is that it is likely too slow for large $p$, large $n$ settings. \texttt{RStan} provides variational Bayes approximations, but we experienced that both the mean-field and the (Gaussian) full-rank approximations do not provide satisfactory results (compared to sampling) for our model. Alternatively, the \texttt{RStan} manual suggests to use a thin QR-decomposition of the design matrix for large scale regression problems, but this did not speed up computations in our setting, probably due to the non-exchangeable prior for the regression coefficients. Dedicated sampling or variational Bayes techniques such as proposed for the horseshoe prior in \cite{makalic2015simple} and \cite{busatto2023informative} may provide more scalable alternatives, but this requires more research.

While certainly simpler than many machine learners, the interpretation of a regression model with many two-way interactions is not trivial \cite[]{afshartous2011key}. Inference for parameters and variable importance scores aids such interpretation.  Therefore, we showed that our methods provides a good basis for inference, which we extend to variable importance scores, in particular the personalized unit change and the Shapley value. For the interventional version of the latter, we derived the computational efficient formula \eqref{phipmain}.  If one prefers to account for dependency when marginalizing, approaches as in \cite[]{aas2021explaining} may be explored, but likely at a high computational cost. For other regressions, e.g. logistic or Cox, formula \eqref{phipmain} is only valid if one evaluates the prediction on the level of the linear predictor. This may be reasonable as the regression coefficients are also on that scale.

While it is straightforward to define global variable importance scores from the personalized ones, including Shapley, it less obvious how to perform inference for these, both from a technical and philosophical perspective. As for the first: one may average absolute or squared scores, but the result lacks a natural null. As for the second: in such models, the importance of a covariate depends on the values of the other ones, and is hence individual specific. Possibly, an informal argument, such as: the credible intervals should not contain `0' for at 10\% of the individuals, may be reasonable in practice, but this needs further research.

Finally, we illustrated that the addition of two-way interactions, as in \texttt{bayint}, may improve predictive performance with respect to a main effect only regression model. In fact, \texttt{bayint} can be very competitive to more advance machine learners, such as random forest. Of course, such comparative results depend strongly on the data, sample size and true complexity of the associations between covariates and response. If prediction is an important aim of the study, possibly alongside interpretation, we recommend comparing the predictive performance of \texttt{bayint} with more flexible machine learners. In the eventual case of inferior predictive performance of the former, one should balance this loss against the improved interpretability.

All-in-all, \texttt{bayint} is an attractive alternative for existing strategies to handle interactions in epidemiological and/or clinical studies, as its linked local shrinkage can improve parameter accuracy, provides appropriate inference and interpretation, and may compete well with less interpretable machine learners in terms of prediction.

\section{Acknowledgements}
The Amsterdam University Medical Centres and the Public Health Service of Amsterdam (GGD Amsterdam) provided core financial support for HELIUS. The HELIUS study is also funded by research grants of the Dutch Heart Foundation (Hartstichting; grant no. 2010T084), the Netherlands Organization for Health Research and Development (ZonMw; grant no. 200500003), the European Integration Fund (EIF; grant no. 2013EIF013) and the European Union (Seventh Framework Programme, FP-7; grant no. 278901).

\bibliography{ShrinkInteractions_vdWiel.bbl}

\newpage
\section{Figures}

\begin{figure}[h]
\begin{center}
\includegraphics[scale=0.44]{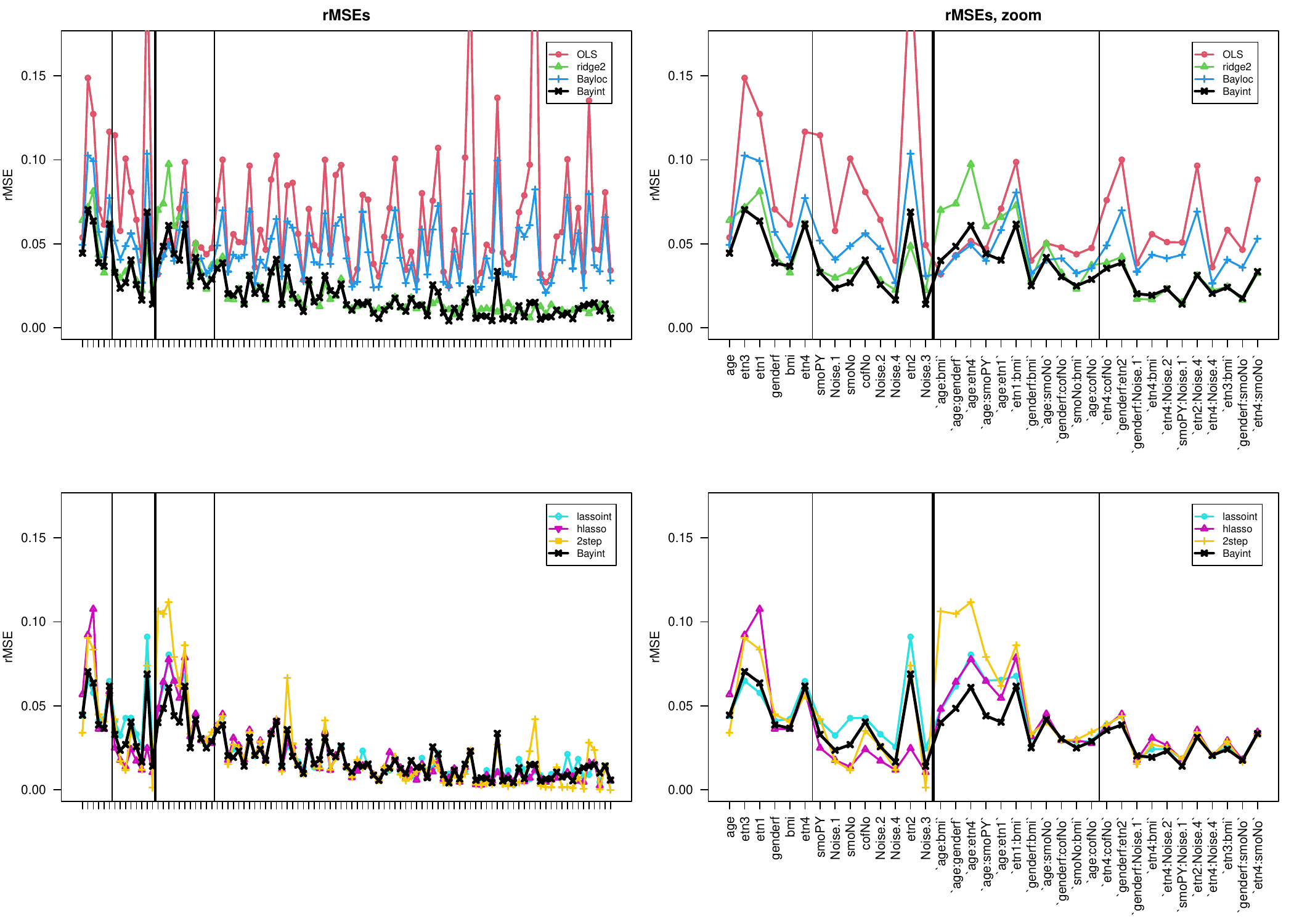}
\end{center}
\caption{rMSEs for 14 main effects and 85 interactions (before and after bold vertical line), each ordered by significance in master set. Thin vertical line demarcates effects significant and non-significant effects in the master set (p$<$.01). Upper and lower plots compare \texttt{Bayint} with methods that do not and do target selection, respectively. Right plots zoom in on main effects and most relevant interactions. \textbf{Cholesterol} as outcome.}\label{rmsecholsplit}
\end{figure}

\newpage
\begin{figure}[h]
\begin{center}
\includegraphics[scale=0.44]{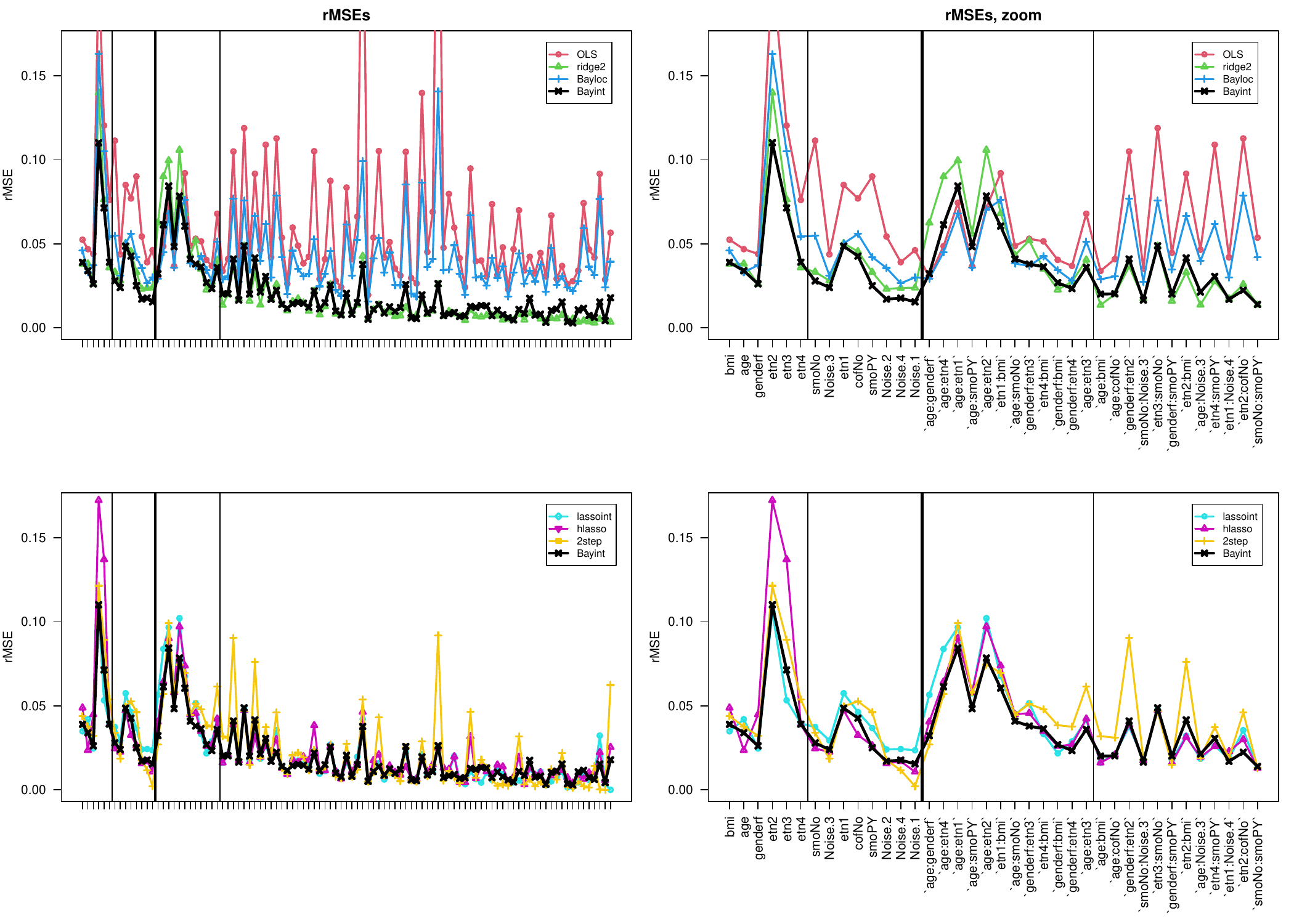}
\end{center}
\caption{rMSEs for 14 main effects and 85 interactions (before and after bold vertical line), each ordered by significance in master set. Thin vertical line demarcates effects significant and non-significant effects in the master set (p$<$.01). Upper and lower plots compare \texttt{Bayint} with methods that do not and do target selection, respectively. Right plots zoom in on main effects and most relevant interactions. \textbf{SBP} as outcome.}\label{rmsesbpsplit}
\end{figure}

\newpage
\begin{figure}[h]
\begin{center}
\includegraphics[scale=0.5]{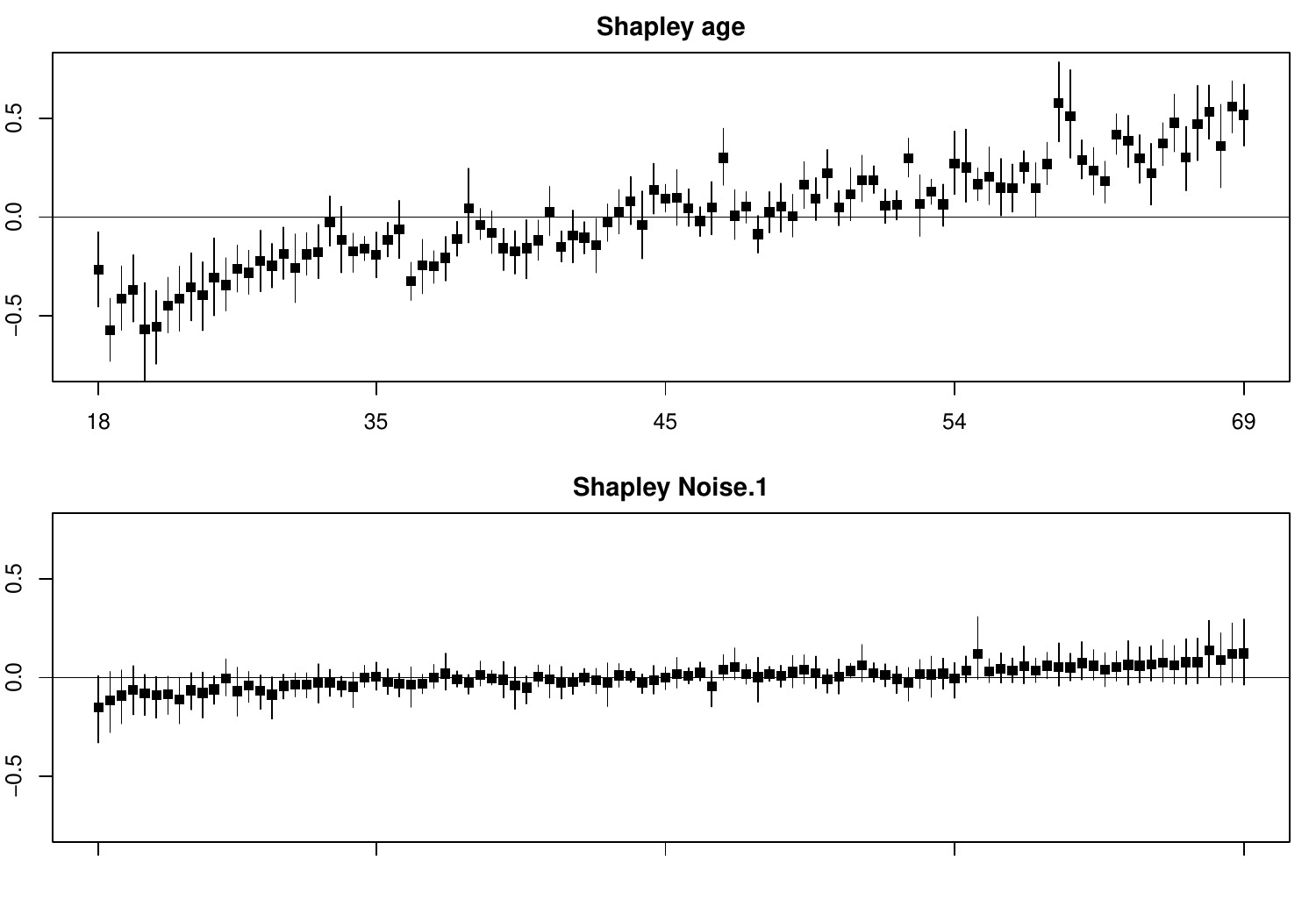}
\end{center}
\caption{Shapley values of `age' and `Noise.1' and their posteriors for 100 random test individuals, ordered by `age' (original scale) and `Noise.1', respectively. \textbf{Cholesterol} (standardized) as outcome.}\label{shapage}
\end{figure}

\newpage
\begin{figure}[h]
\begin{center}
\includegraphics[scale=0.35]{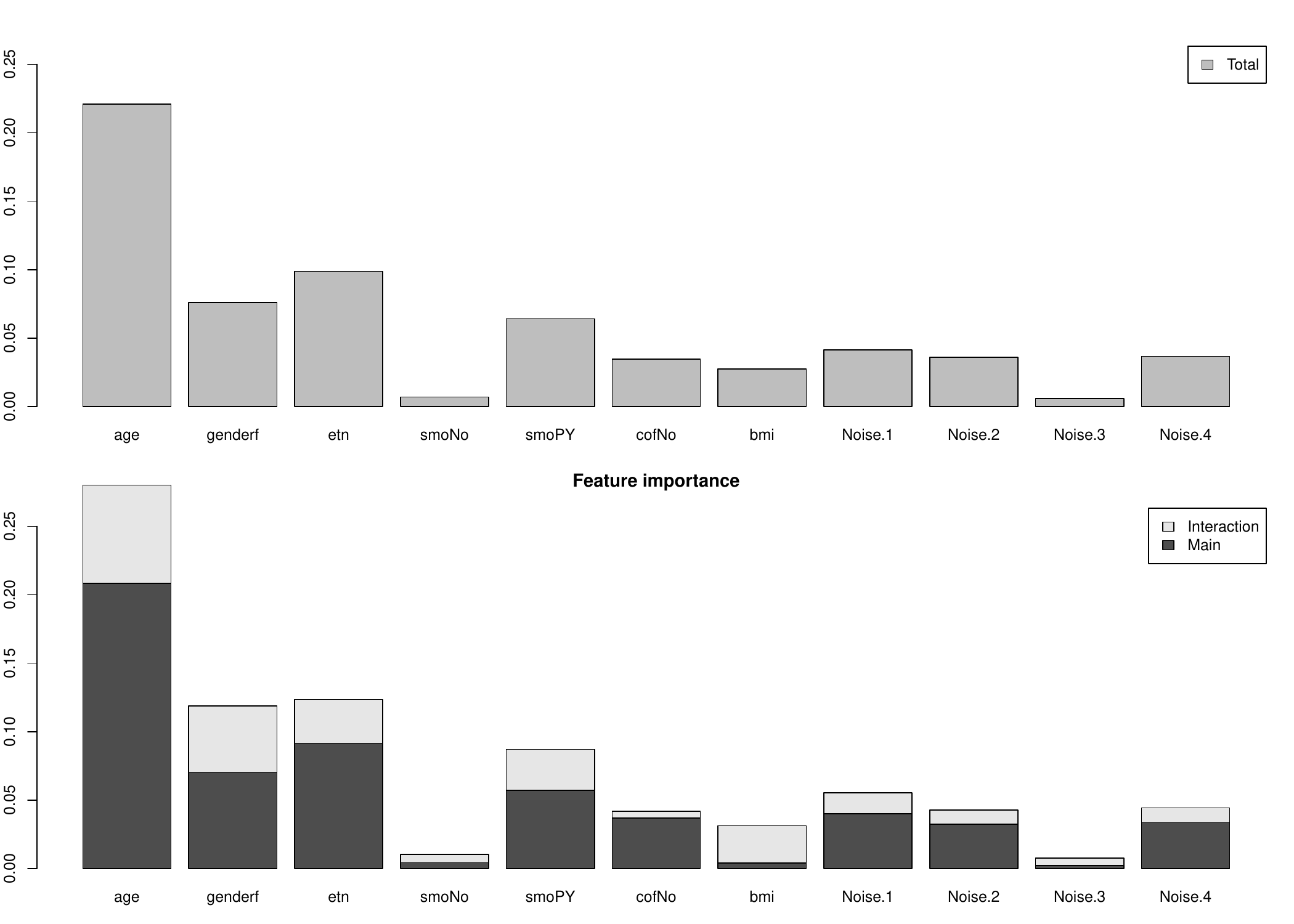}
\end{center}
\caption{Global variable importance scores. Top: mean absolute Shapley values ($I_j$); Bottom: mean absolute contributions of main effect ($I_j^{\text{main}}$) and interactions ($I_j^{\text{int}}$). \textbf{Cholesterol} as outcome.}\label{shapimportance}
\end{figure}

\newpage
\begin{figure}[h]
\begin{center}
\includegraphics[scale=0.4]{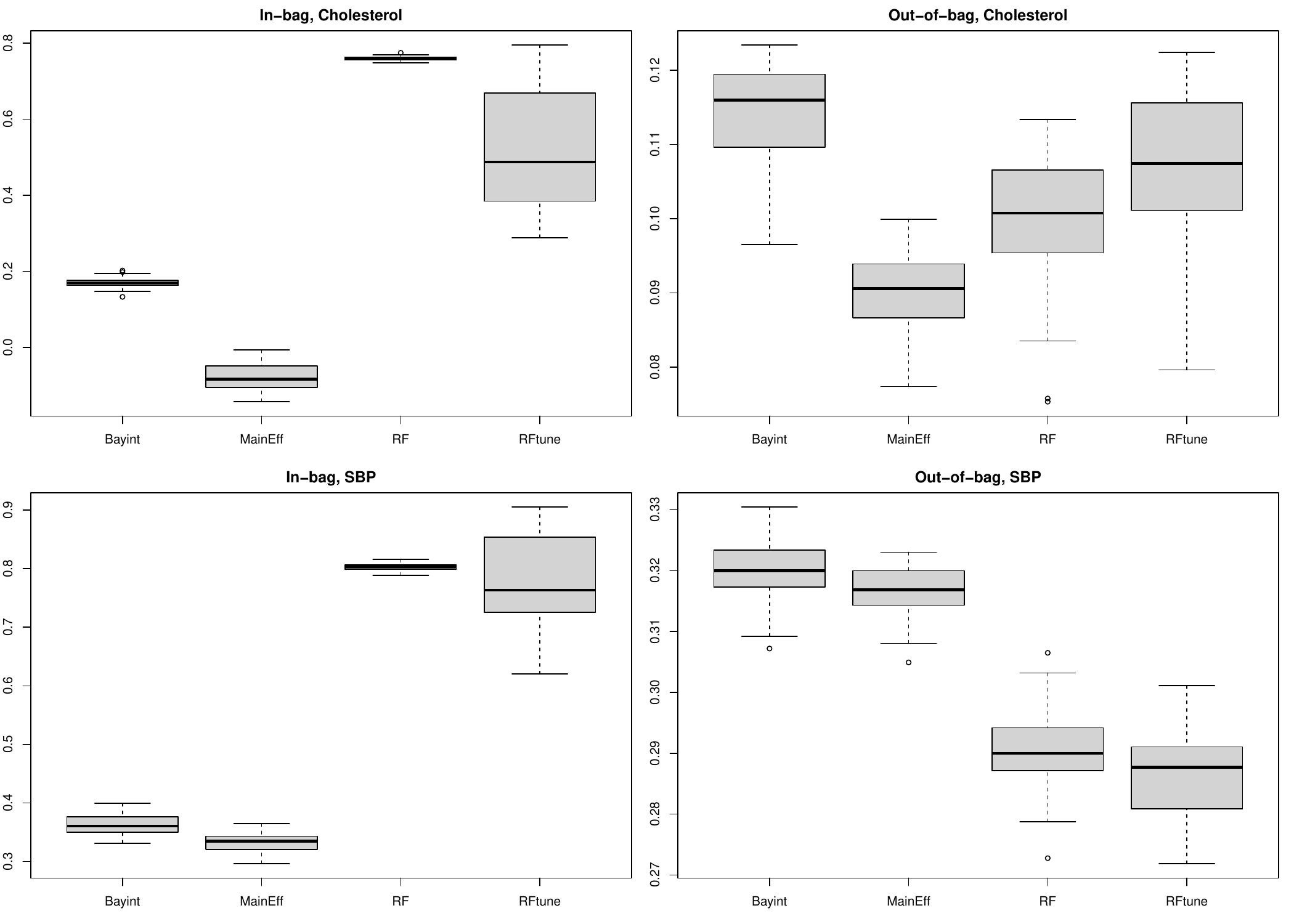}
\end{center}
\caption{In-bag and out-of-bag $R^2$s for 25 training sets of size n=1,000 for cholesterol (top) and SBP (bottom) as outcome. Methods: \texttt{Bayint}: Bayesian linked shrinkage model; \texttt{MainEff}: OLS with main effects only; \texttt{RF} (\texttt{RFtune}): Random Forest with default (tuned) parameters.
}\label{r2s}
\end{figure}

\newpage

\section{Supplementary Material}
\setcounter{figure}{0}
\subsection{Implementation of alternative methods}
Below we give details on the implementation of each of the alternative methods used in the manuscript. Code is available on:
\url{https://github.com/markvdwiel/ThinkInteractions/}.
\begin{itemize}
   \item \texttt{OLS} and \texttt{2step} are fit using the base \texttt{R lm} function, where the latter first fits a `main effect only model', and then fits a new model with only the significant main effects ($p<0.05$) and their interactions
  \item \texttt{ridge2} is fit using \texttt{mgcv} (v 1.8-42) \cite[]{wood2011fast}, which automatically tunes the two penalties
  \item \texttt{lassoint} is fit using \texttt{glmnet} (v 4.1-7) using the penalty that minimises the 10-fold cross-validated squared prediction error (default)
\item \texttt{hlasso} is fit using \texttt{glinternet} (v 1.0.12) with the penalty cross-validated as for \texttt{lassoint} (default)

  \item Variations of our model,  \texttt{Bayintadd}, \texttt{Bay0int} and \texttt{Bayint*}, as well as \texttt{Bayloc} are implemented in \texttt{RStan} (v 2.21.8), using the exact same priors and sampling scheme as used for our model (\texttt{Bayint})
 \item \texttt{RF} is either fit using the defaults of the \texttt{rfsrc} function in the \texttt{randomForestSRC} package (v 3.2.2) or with hyperparameters (\texttt{mtry} and \texttt{nodesize}) tuned for optimal predictive performance using the \texttt{tune.rfsrc} function
\end{itemize}


\subsection{Shapley for regression model with two-way interactions}
Here, we derive the interventional Shapley value for our model. For clarity we fix the sample index $i$, and denote $x_j = x_{ij}$ and its realization by
$x^*_j$. W.l.o.g. we assume index $p$ for the covariate of which we compute the Shapley value. Denote by $\mcS \subseteq \mathcal{P} = \{1, \ldots, p\}$ the set of players, i.e. the covariate indices for which their value is set at the realization, and its complement
$\mcS'$: the set of non-players over which the predictions are marginalized.
Define weights $w(\mcS) = \frac{1}{p}[\binom{p-1}{|\mcS|}]^{-1}.$
Write $\bm{x} = (x_j)_{j=1}^p, \bm{x}_{\mcS} = (x_j)_{j \in  \mcS}, \bm{x}^*_\mcS = (x^*_j)_{j \in  \mcS},$ and  $f(\bm{x})$ for the predicted value.
Moreover, let $\nu(\mcS) = E[f(\bm{x})|\bm{x}_\mcS = \bm{x}^*_\mcS]$, which is the expected predicted value, with the expectation computed over all covariates $j \in \mcS'$.
Then,
\begin{equation}\label{phidef}
\phi_p =  \sum_{\mcS \in \mathcal{P} \setminus \{p\}} w(\mcS) \biggl(\nu(\mcS \cup \{p\}) - \nu(\mcS)\biggr).
\end{equation}
For the interventional Shapley, it is assumed that $P(\bm{x}_{\mcS'} | \bm{x}_{\mcS} = \bm{x}^*_\mcS) = P(\bm{x}_{\mcS'}).$ Then,
$\nu(\mcS) = E_{\bm{x}_{\mcS'}}[f(\bm{x}^*_\mcS,\bm{x}_{\mcS'})]$.

\para
The following result shows that for a linear regression model with all main effects and two-way interactions the exact $\phi_p$ can be computed with linear complexity $\mathcal{O}(p)$ instead of exponential. For simplicity, we first assume categorical variables are absent.

\para
\begin{theorem}
Let $\phi_p$ be defined as in \eqref{phidef} and assume $P(\bm{x}_{\mcS'} | \bm{x}_{\mcS} = \bm{x}^*_\mcS) = P(\bm{x}_{\mcS'}).$ Then for $f(\bm{x}) =
\alpha_0 + \sum_{j=1}^p \beta_j x_j + \sum_{k=2}^p \sum_{j=1}^{k-1} \beta_{jk} x_j x_k$ we have:
\begin{equation}\label{phip}
\phi_p = \beta_p(x_p^* - E[x_p]) + \frac{1}{2}\biggl(\sum_{j=1}^{p-1}\beta_{jp}(E[x_j]x_p^*  - E[x_jx_p]) + \sum_{j=1}^{p-1}\beta_{jp}(x_j^*x_p^* - x_j^* E[x_p])\biggr).
 \end{equation}
\end{theorem}
\begin{proof}
First, for convenience we write $\sum_{k=2}^p \sum_{j=1}^{k-1} \beta_{jk} x_j x_k = \frac{1}{2}\sum_{j,k} \beta_{jk} x_j x_k,$ with $\beta_{kj}=\beta_{jk}$ and $\beta_{jj} = 0$. Then, decompose $\nu(\mcS)$ into the main effects term (including the intercept) and an interaction term:
$\nu(\mcS) = \nu_m(\mcS) + \nu_i(\mcS).$ Analogously, write $\phi_p = \phi^m_p + \phi^i_p.$ \cite{aas2021explaining} show that $\phi^m_p = \beta_p(x_p^* - E[x_p])$.
Hence, we focus on the contribution of the interactions, which equals:
\begin{equation*}
\nu_i(\mcS) =  \frac{1}{2}\bigg(\sum_{\substack{j \in \mcS'\\k \in \mcS'}} \beta_{jk} E[x_jx_k] +  2\sum_{\substack{j \in \mcS \\ k \in \mcS'}} \beta_{jk}  x_j^* E[x_k] + \sum_{\substack{j \in \mcS \\ k \in \mcS}} \beta_{jk} x_j^* x_k^*\biggr).
\end{equation*}
Likewise,
\begin{equation*}
\nu_i(\mcS \cup \{p\}) =  \frac{1}{2}\bigg(\sum_{\substack{j \in \mcS' \setminus \{p\}\\k \in \mcS'\setminus \{p\}}} \beta_{jk} E[x_jx_k] +  2\sum_{\substack{j \in \mcS \cup \{p\}\\k \in  \mcS'\setminus \{p\}}} \beta_{jk}  x_j^* E[x_k] + \sum_{\substack{j \in \mcS \cup \{p\}\\k \in \mcS \cup \{p\}}} \beta_{jk} x_j^* x_k^*\biggr).
\end{equation*}
Therefore, we have:
\begin{equation}\label{nudif}
\begin{split}
\nu_i(\mcS &\cup \{p\}) - \nu_i(\mcS)\\
&=  -\sum_{j \in \mcS' \setminus \{p\}} \beta_{jp} E[x_jx_p] + \sum_{j \in \mcS' \setminus \{p\}} \beta_{jp} E[x_j]x_p^* -
 \sum_{j \in \mcS}\beta_{jp} x_j^* E[x_p] + \sum_{j \in \mcS}\beta_{jp} x_j^* x_p^* \\
 &= \sum_{j \in \mcS' \setminus \{p\}}T_j^1 +  \sum_{j \in \mcS}T_j^2,
 \end{split}
\end{equation}
with $T_j^1 =  \beta_{jp} (E[x_j]x_p^* -  E[x_jx_p])$ and $T_j^2 =  \beta_{jp} (x_j^*x_p^* -  x_j^*E[x_p])$.
Then, $$\phi^i_p = \sum_{\mcS} w(\mcS)\biggl( \sum_{j \in \mcS' \setminus \{p\}} T_j^1 +  \sum_{j \in \mcS} T_j^2\biggr).$$
To compute $\phi^i_p$, we have for the second term:
$$\sum_{\mcS} w(\mcS) \sum_{j \in \mcS} T_j^2 = \sum_{\mcS} w(\mcS) \sum_{j =1}^{p-1} I(j \in \mcS) T_j^2 = \sum_{k=1}^{p-1} w_\ell \binom{p-2}{k-1}\sum_{j=1}^{p-1} T_j^2,$$
as the weight $w(\mcS) = w_\ell = 1/p[\binom{p-1}{\ell}]^{-1},$ when $|\mcS| = \ell$, and $\binom{p-2}{\ell-1}$ counts the number of size $\ell$ subsets out of $\{1, \ldots, p-1\}$ that necessarily contain $j$, which equals the number of size $\ell-1$ subsets out of $\{1, \ldots, p-1\} \setminus \{j\}$.
Likewise for the first term:
$$\sum_{\mcS} w(\mcS) \sum_{j \in \mcS' \setminus \{p\}} T_j^1 = \sum_{\mcS} w(\mcS) \sum_{j=1}^{p-1} I(j \in \mcS' \setminus \{p\}) T_j^1 = \sum_{\ell=0}^{p-2} w_\ell \binom{p-2}{p-2-\ell}\sum_{j=1}^{p-1} T_j^1,$$ as $|\mcS| = \ell$ implies $|\mcS' \setminus \{p\}| = p-1-\ell$ and  $\binom{p-2}{p-2-\ell}$ counts the number of size $p-1-\ell$ subsets out of $\{1, \ldots, p-1\}$ that necessarily contain $j$. Finally, note that $w_\ell \binom{p-2}{\ell-1} = (1/p) \binom{p-2}{\ell-1}/\binom{p-1}{\ell} = \frac{\ell}{p(p-1)},$ and $w_\ell \binom{p-2}{p-2-\ell} = \frac{p-1-\ell}{p(p-1)}$ and that $\sum_{\ell=0}^{p-2} (p-2-\ell) = \sum_{\ell=1}^{p-1} \ell = p(p-1)/2$. Substitution completes the proof.
\end{proof}

\begin{cor}
If covariates are centered, we have:
\begin{equation}\label{phip}
\phi_p = \beta_p x_p^* + \frac{1}{2}\biggl(\sum_{j=1}^{p-1}\beta_{jp} x_j^*x_p^* - \sum_{j=1}^{p-1}\beta_{jp} E[x_jx_p]\biggr).
 \end{equation}
 \end{cor}
\begin{proof} This simply follows from the fact that centering implies $E[x_j] = 0$.
\end{proof}

\begin{cor}
If a categorical covariate is present, then:
\begin{enumerate}
\item For a non-categorical covariate $\phi_p$ remains unchanged
\item For a categorical covariate $\tilde{p}$ with $Q$ levels, corresponding to indices $\{p+1 - Q, \ldots, p\}$,
\begin{equation}\label{phip2}
\phi_{\tilde{p}} = \sum_{q=1}^Q \phi_{p+1-q}.
 \end{equation}
 \end{enumerate}
 \end{cor}
\begin{proof}
First, \emph{1.} follows directly from redefining $\mathcal{P} = 1, \ldots, p^{-}$, where $p^{-}$ is the number of covariates, counting
categorical variables as one. When a categorical variable $\tilde{p}$ is in $\mcS$ all its modeling components contribute to $\nu(\mcS)$. Then, replacing $p$ by $p^{-}$ in the weights $w$ and in the combinatorial argument above, renders exactly the same result for a non-categorical covariate.
For \emph{2.}, simply note that $\nu(\mcS \cup \{\tilde{p}\})$ extends linearly over all regression terms with indices $p+1 - Q, \ldots, p$, when excluding interaction terms between levels of the $\tilde{p}$ in the regression model, as is the convention.
\end{proof}

\subsection{Supplementary Figures}


\begin{figure}[h]
\begin{center}
\includegraphics[scale=0.35]{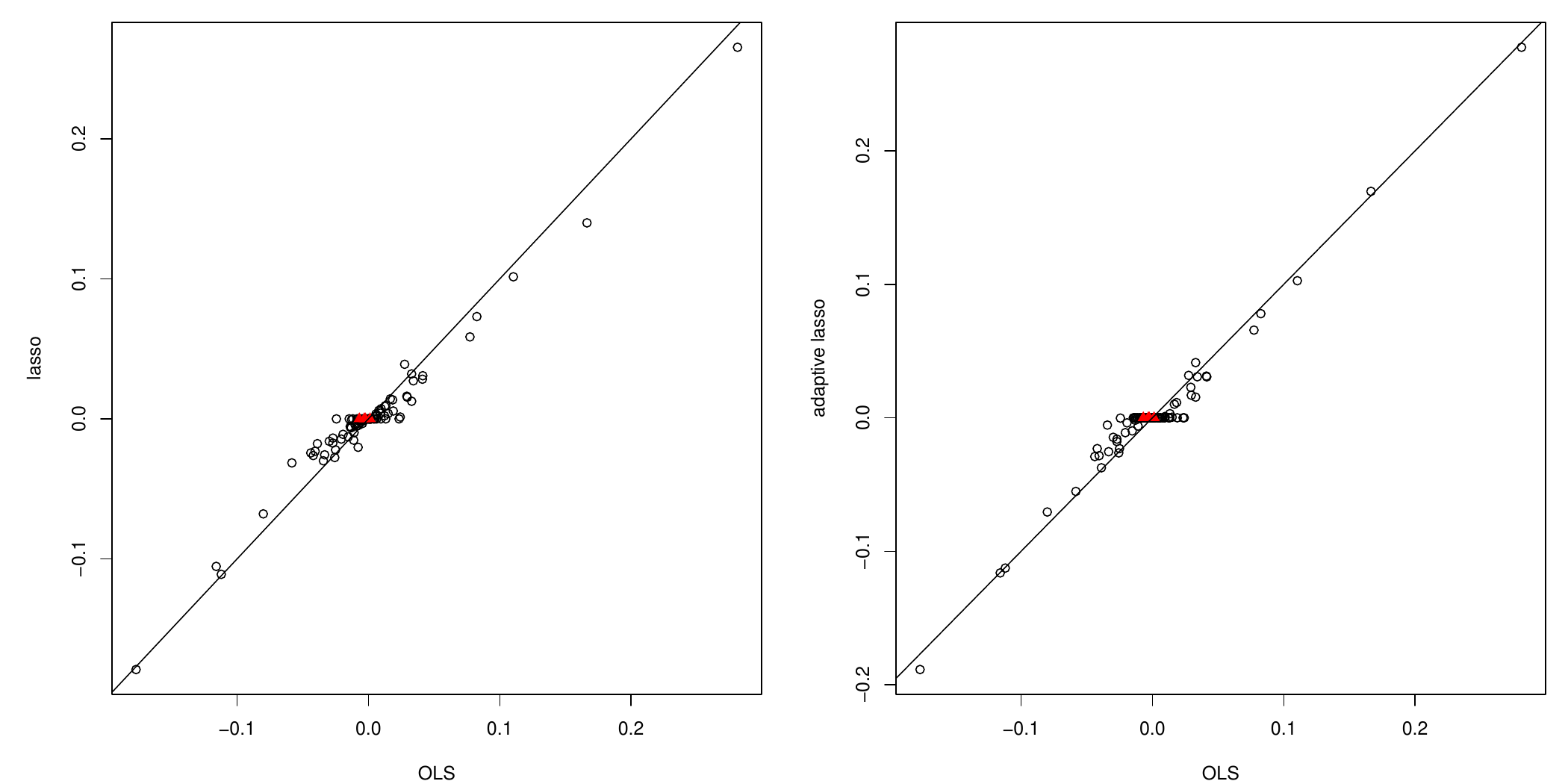}
\end{center}
\caption{OLS estimates (x-axis) versus (adaptive)  lasso estimates (y-axis) on the master set ($N=21,570$). Red triangles: estimates for Noise main effects.
Adaptive lasso uses OLS-based penalty weights to de-bias estimates of large coefficients. \textbf{Cholesterol} as outcome.}\label{olsvslassochol}
\end{figure}

\begin{figure}[h]
\begin{center}
\includegraphics[scale=0.35]{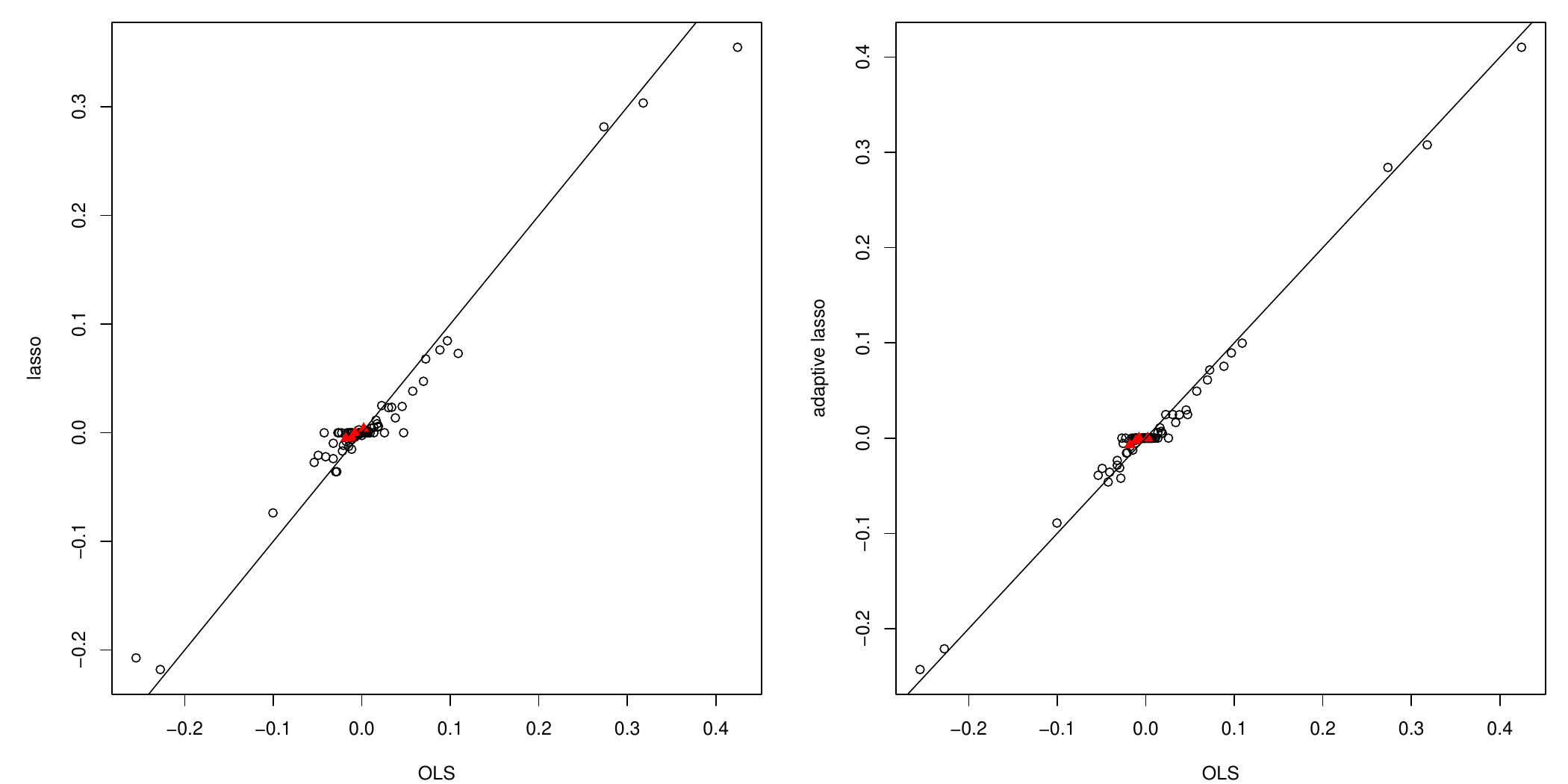}
\end{center}
\caption{OLS estimates (x-axis) versus (adaptive)  lasso estimates (y-axis) on the master set ($N=21,570$). Red triangles: estimates for Noise main effects.
Adaptive lasso uses OLS-based penalty weights to de-bias estimates of large coefficients. \textbf{SBP} as outcome.}\label{olsvslassosbp}
\end{figure}

\begin{figure}[h]
\begin{center}
\includegraphics[scale=0.35]{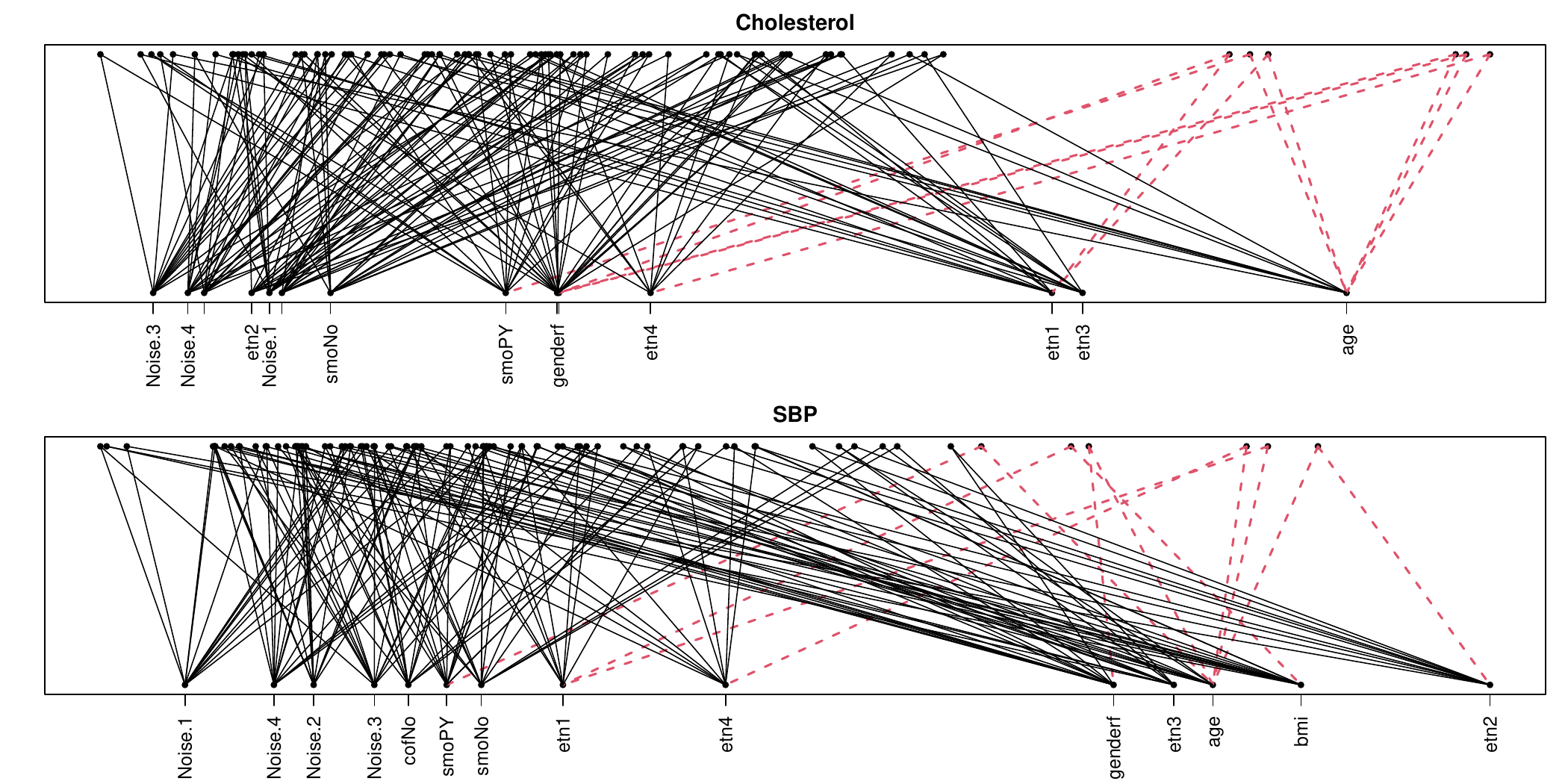}
\end{center}
\caption{Connecting main effects to interactions for Cholesterol model (top) and SBP model (bottom). For each plot, lower line of dots shows square-root of absolute \emph{true} coefficients (estimated from Master set) for main effects, top line for interactions (scaled to match the scale of the main effects). Interactions are connected with their corresponding main effects. Red dashed lines are used for the six largest interaction effects. }\label{mainintconnect}
\end{figure}

\begin{figure}[h]
\begin{center}
\includegraphics[scale=0.38]{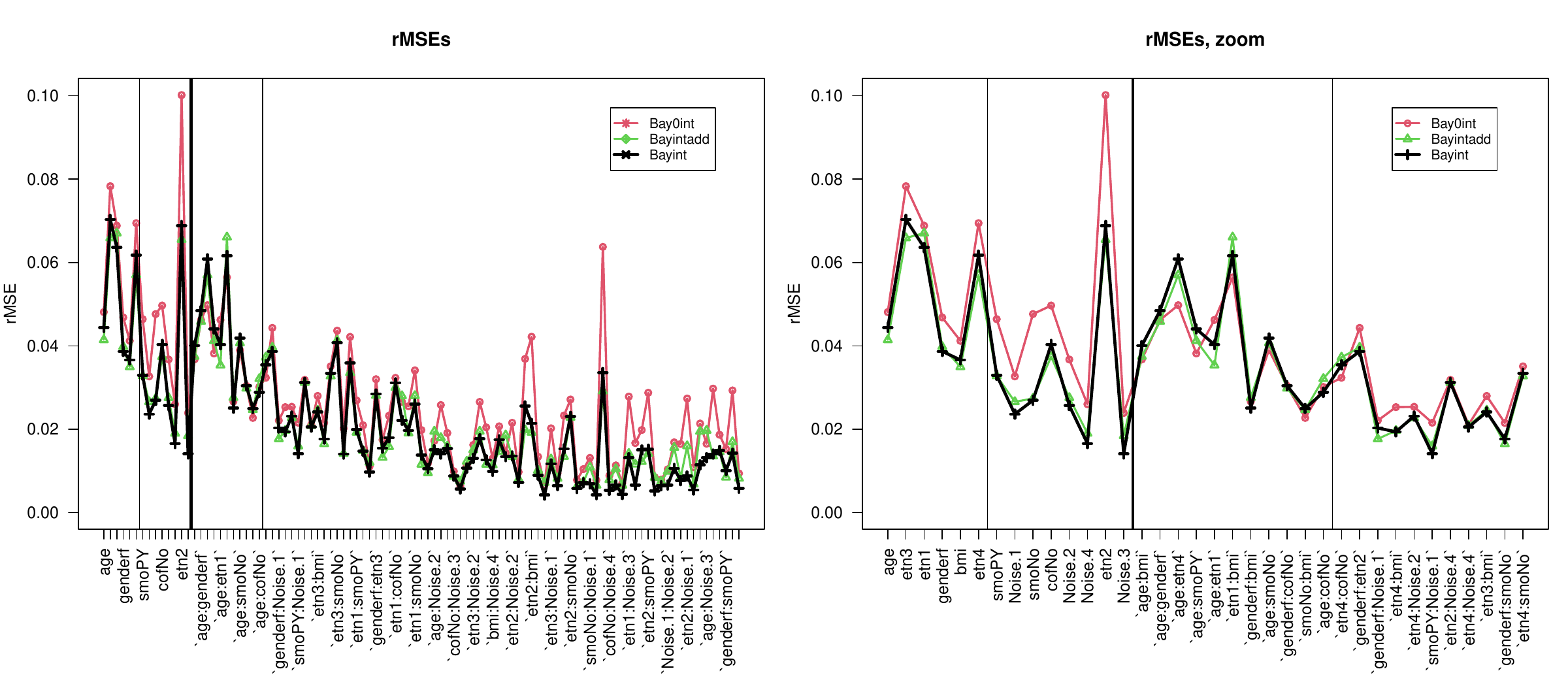}
\end{center}
\caption{Comparing \texttt{Bayint} to \texttt{Bay0int} and \texttt{Bayintadd}. rMSEs for 14 main effects and 85 interactions (before and after bold vertical line), each ordered by significance in master set. Thin vertical line
demarcates effects significant and non-significant effects in the master set (p$<$.01). Right plots zoom in on main effects and most relevant interactions. \textbf{Cholesterol} as outcome.}\label{rmsebayintcompetcholsplit}\end{figure}

\begin{figure}
\begin{center}
\includegraphics[scale=0.5]{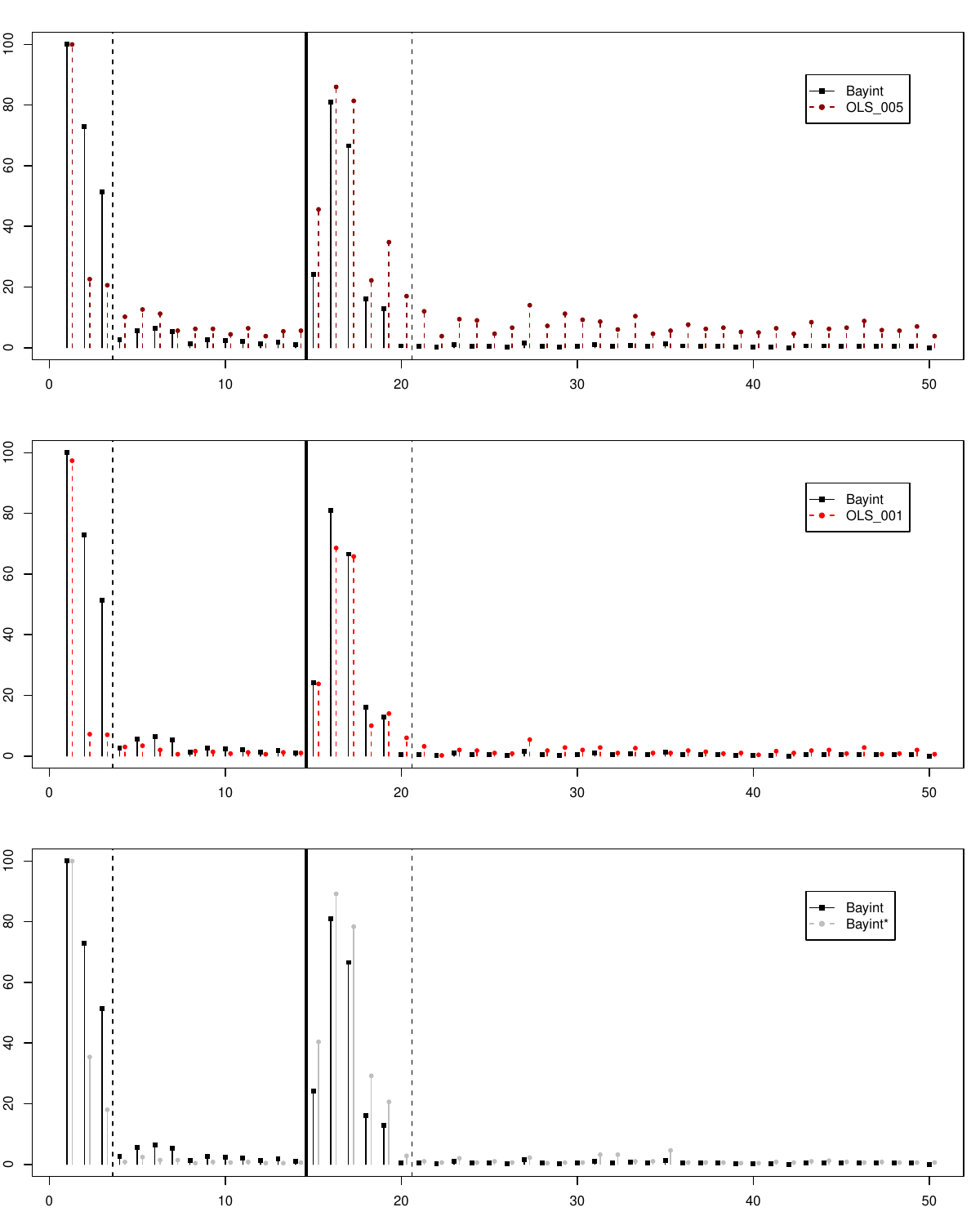}
\end{center}
\caption{Percentage detected effects (out of 500 subsets), with covariates ordered according to absolute effect size in master set. Bold line demarcates main effects and interactions, dashed line separates effect sizes in master set larger and smaller than 0.075. Only first 50 are shown; remainder 49 behave similarly to $25, \ldots, 50$.  \textbf{Cholesterol} as outcome.}\label{detections}
\end{figure}

\begin{figure}
\begin{center}
\includegraphics[scale=0.5]{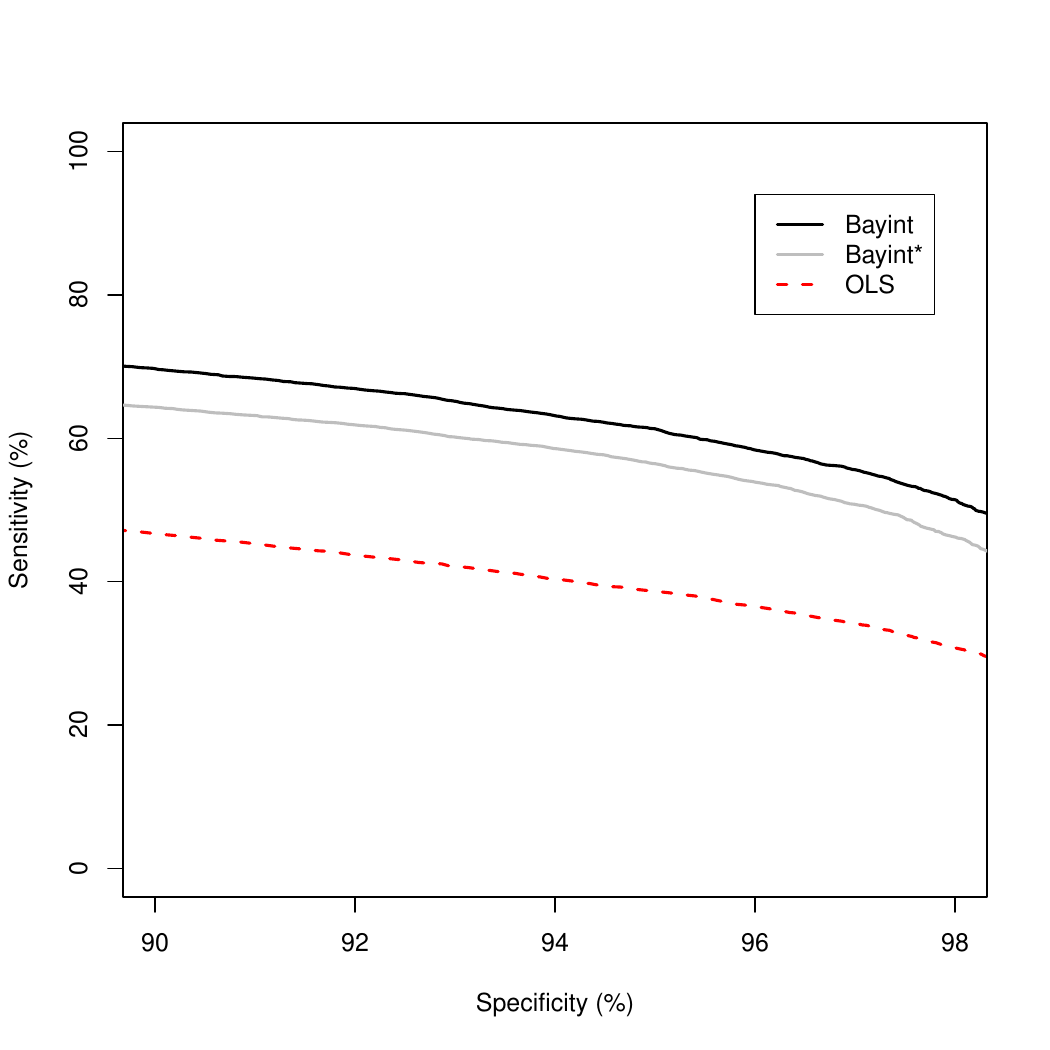}
\end{center}
\caption{Sensitivity versus specificity, averaged over 500 subsets. \textbf{Cholesterol} as outcome.}\label{roc}
\end{figure}

\begin{figure}
\begin{center}
\includegraphics[scale=0.4]{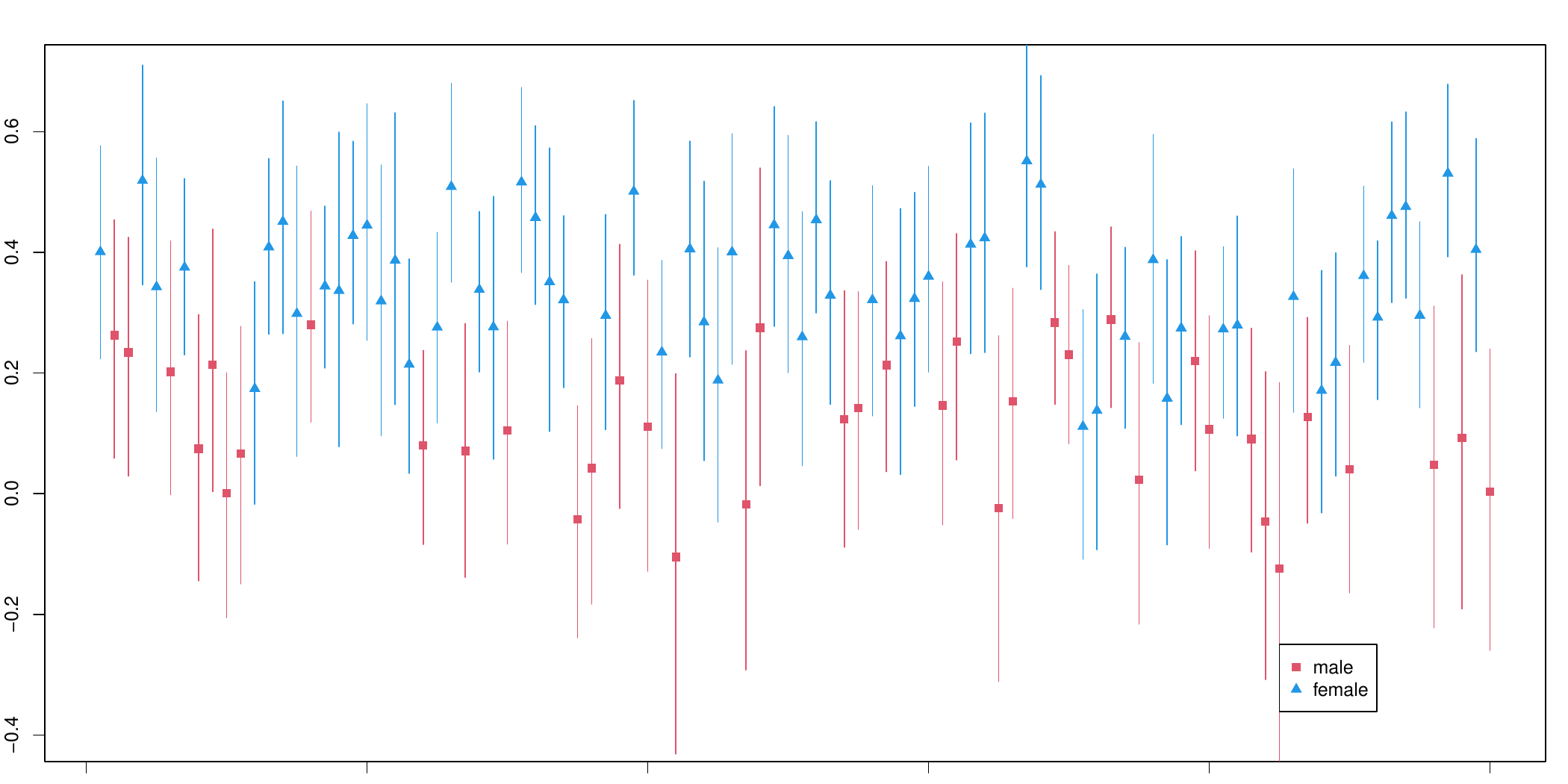}
\end{center}
\caption{Effect of one unit age increase for males and females (100 random test samples, ordered by age) and its uncertainty, accounting for interactions of age with other covariates.  \texttt{Bayint} model fitted on 1,000 training samples. \textbf{Cholesterol} (standardized) as outcome.}\label{deltaage}
\end{figure}

\begin{figure}[h]
\begin{center}
\includegraphics[scale=0.4]{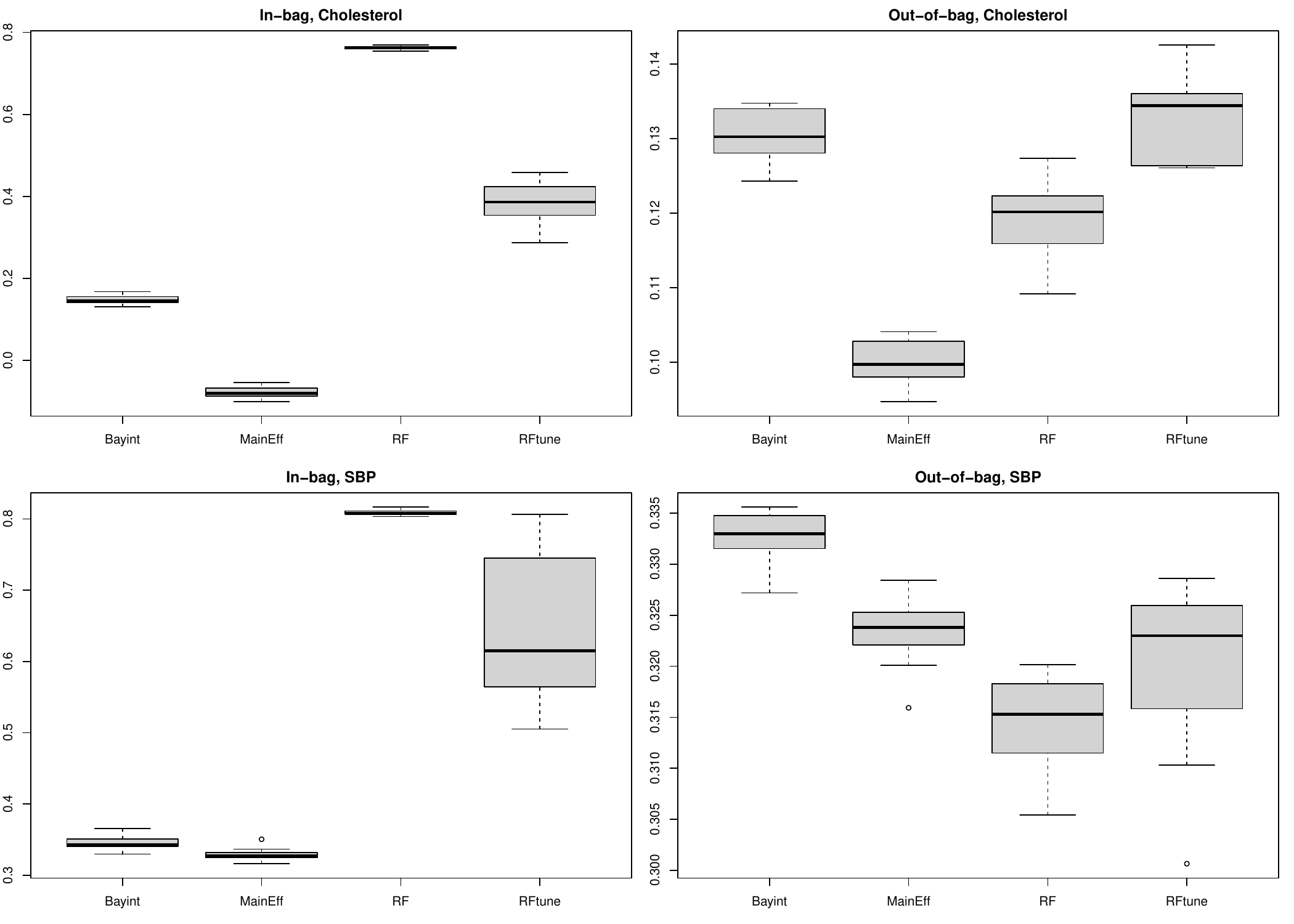}
\end{center}
\caption{In-bag and out-of-bag $R^2$s for 10 training sets of size n=5,000 for cholesterol (top) and SBP (bottom) as outcome. Methods: \texttt{Bayint}: Bayesian linked shrinkage model; \texttt{MainEff}: OLS with main effects only; \texttt{RF} (\texttt{RFtune}): Random Forest with default (tuned) parameters.}\label{r2s5000}
\end{figure}


\begin{figure}[h]
\begin{center}
\includegraphics[scale=0.35]{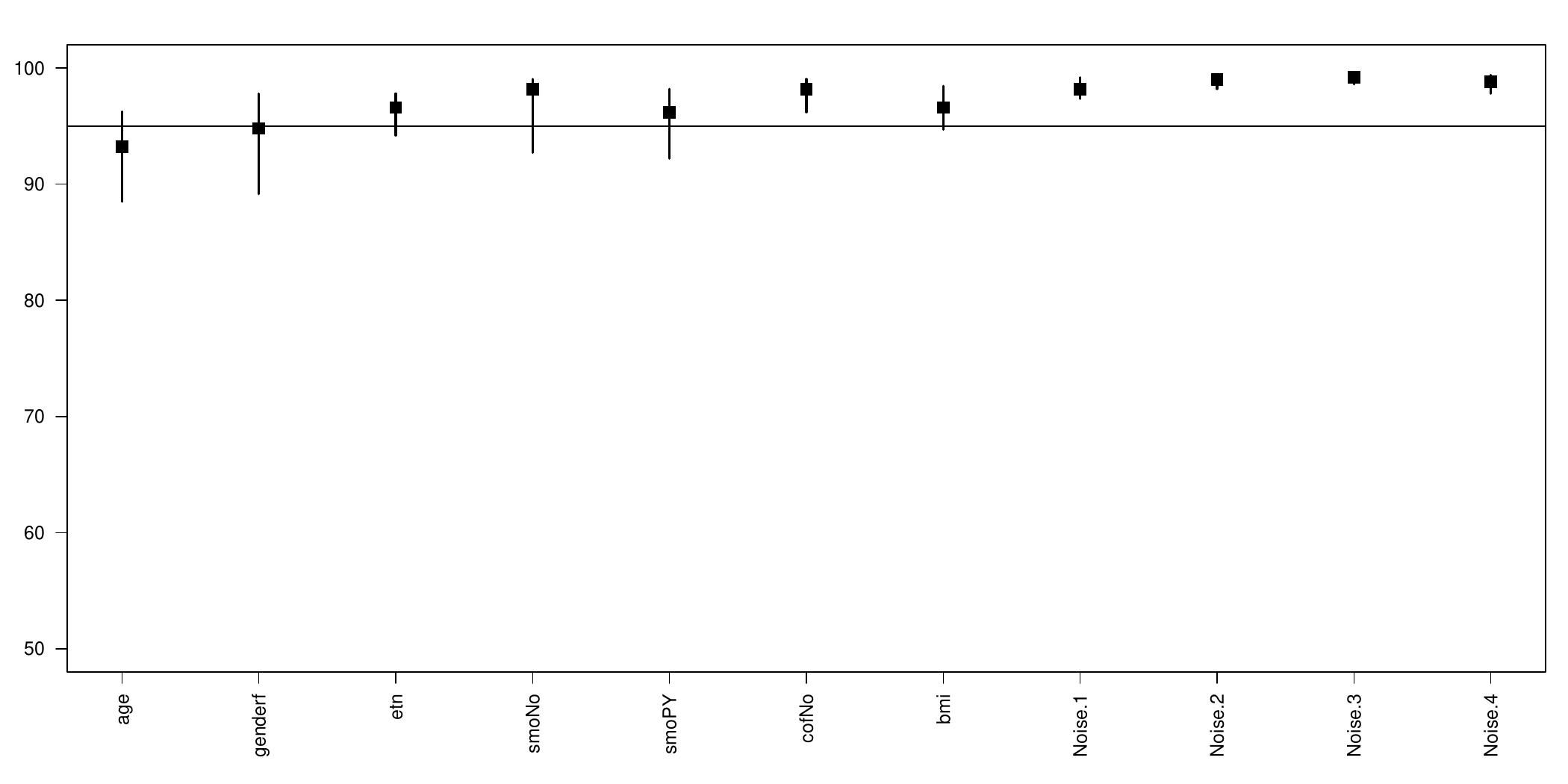}
\end{center}
\caption{Coverages of 95\% credible intervals for Shapley values. Estimated Shapley's and intervals are obtained from 500 training sets.
True Shapley values are based on parameter estimates from the master set. All Shapleys's are computed for 200 random test individuals. For the test set median, first and third
quartile coverage are shown.  \textbf{Cholesterol} as outcome.}\label{coverchol}
\end{figure}

\begin{figure}[h]
\begin{center}
\includegraphics[scale=0.35]{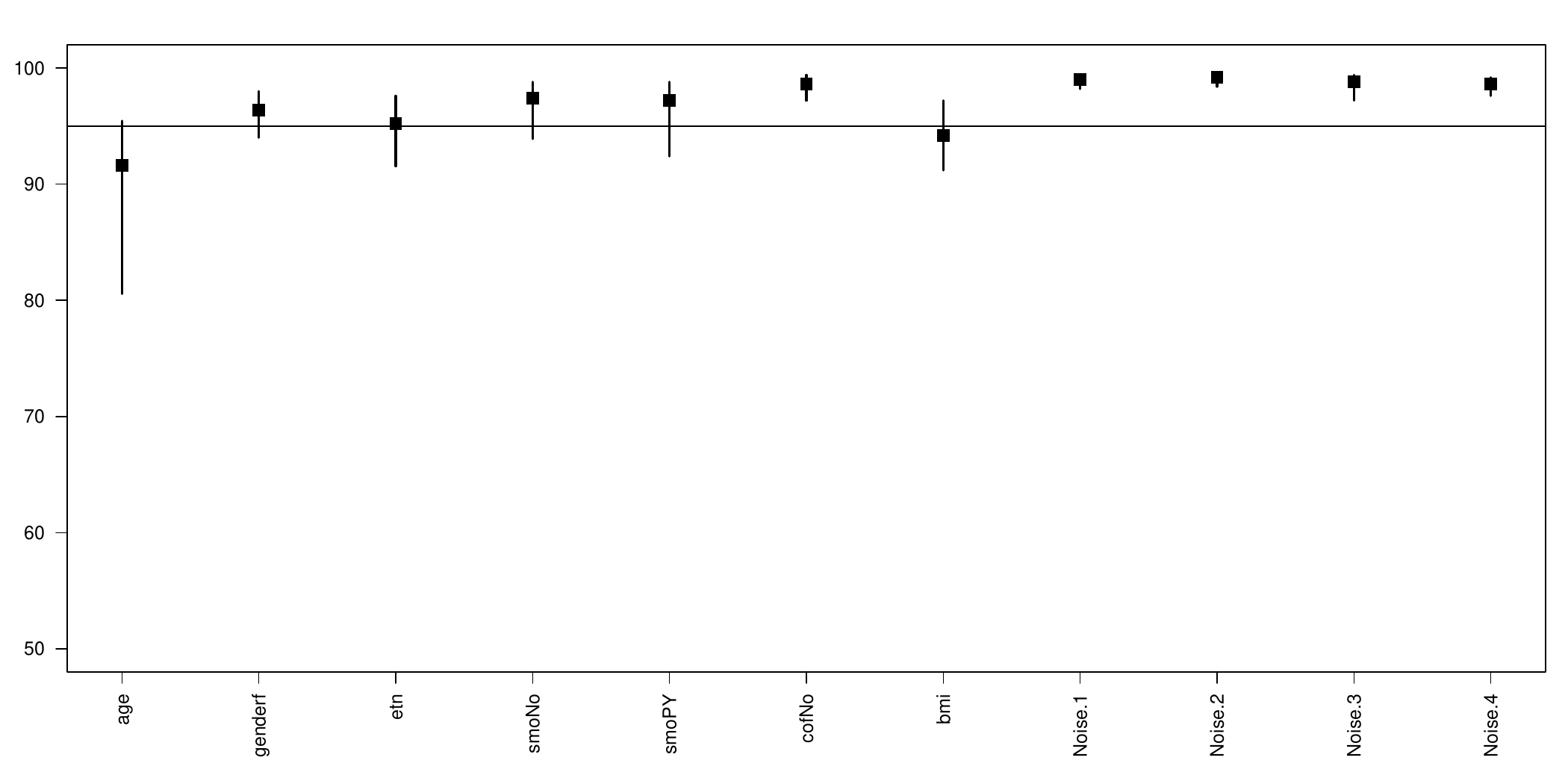}
\end{center}
\caption{Coverages of 95\% credible intervals for Shapley values. Estimated Shapley's and intervals are obtained from 500 training sets.
True Shapley values are based on parameter estimates from the master set. All Shapleys's are computed for 200 random test individuals. For the test set median, first and third
quartile coverage are shown. \textbf{Sbp} as outcome.}\label{coversbp}
\end{figure}

\begin{figure}
\begin{center}
\includegraphics[scale=0.45]{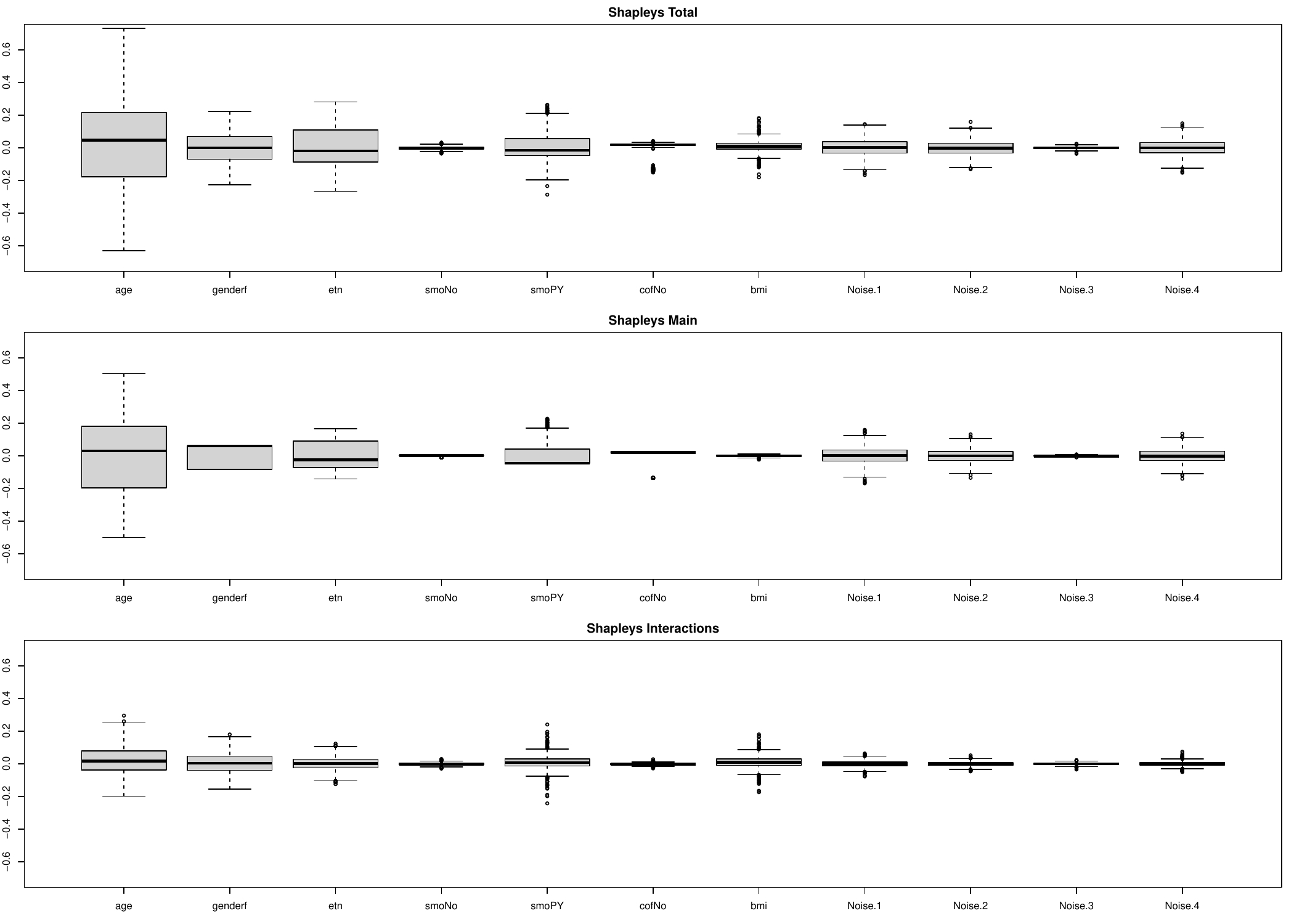}
\end{center}
\caption{Distribution of Shapley values of all covariates over 1,000 random test individuals. \textbf{Cholesterol} (standardized) as outcome.}\label{shapleyall}
\end{figure}

\begin{figure}
\begin{center}
\includegraphics[scale=0.44]{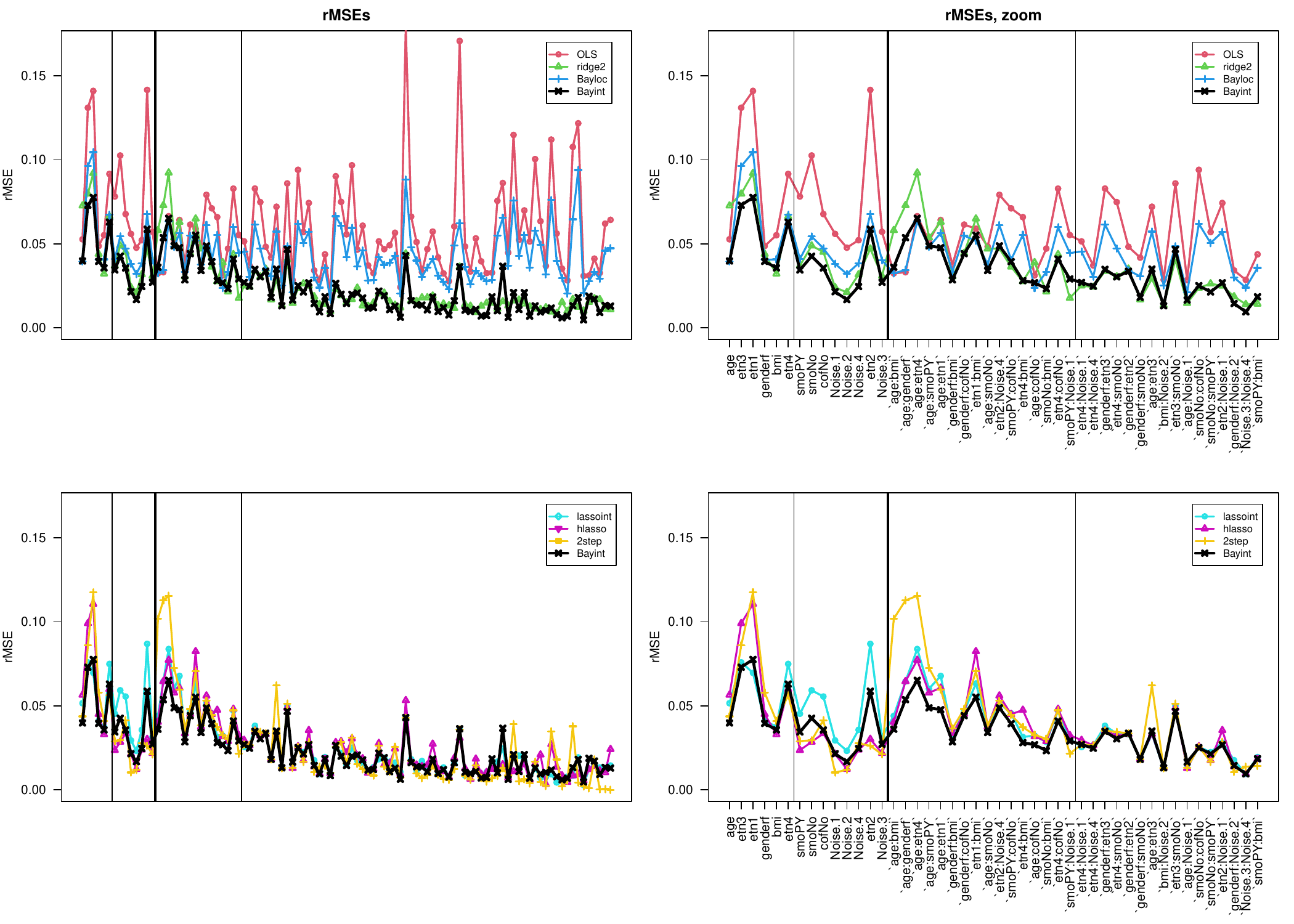}
\end{center}
\caption{rMSEs computed from \emph{synthetic} data set for 14 main effects and 85 interactions (before and after bold vertical line), each ordered by significance in synthetic master set. Thin vertical line demarcates effects significant and non-significant effects in the synthetic master set (p$<$.01). Upper and lower plots compare \texttt{Bayint} with methods that do not and do target selection, respectively. Right plots zoom in on main effects and most relevant interactions. \textbf{Cholesterol} as outcome.}\label{rmsecholsplit_synth}
\end{figure}

\begin{figure}
\begin{center}
\includegraphics[scale=0.44]{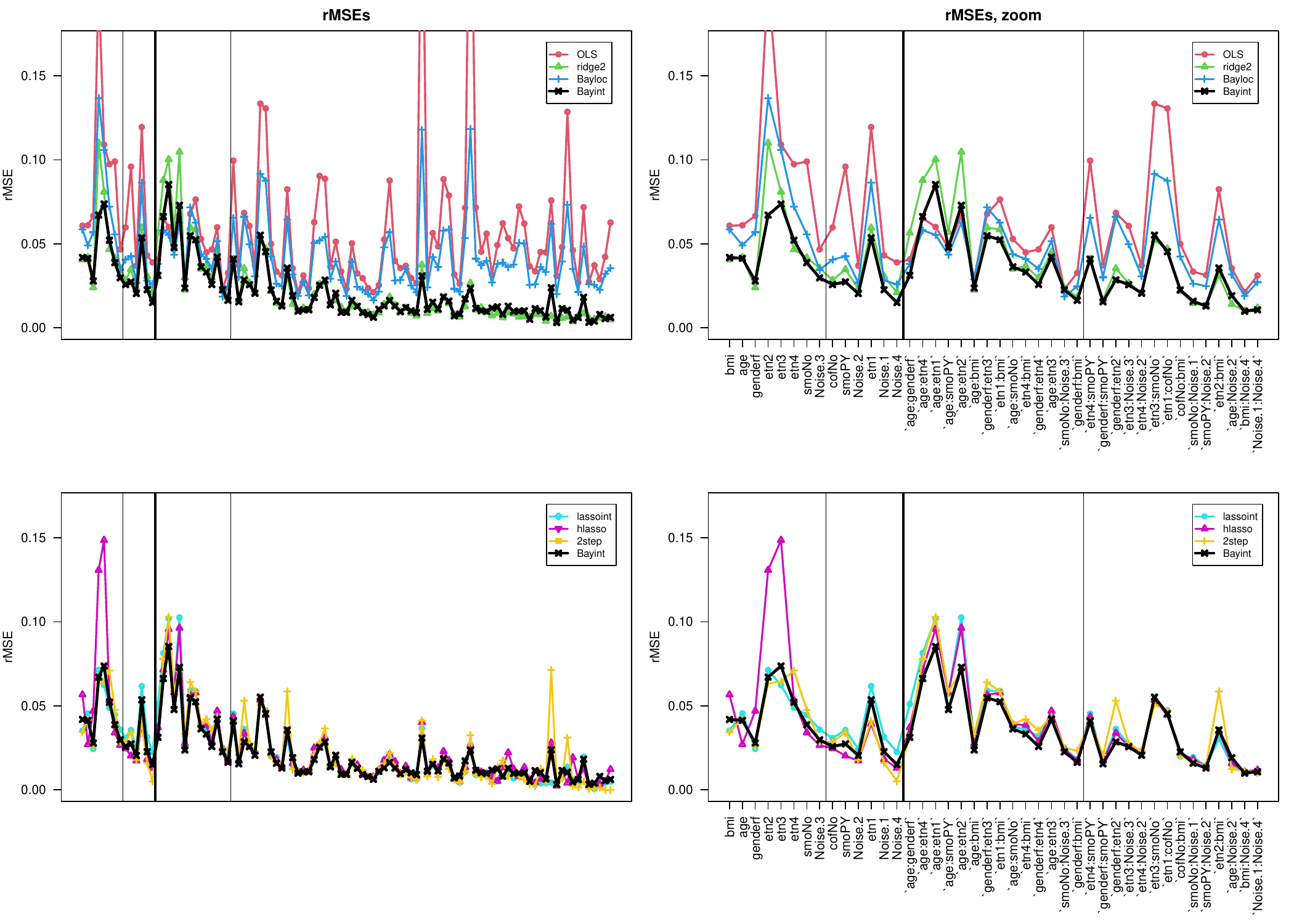}
\end{center}
\caption{rMSEs computed from \emph{synthetic} data set for 14 main effects and 85 interactions (before and after bold vertical line), each ordered by significance in synthetic master set. Thin vertical line demarcates effects significant and non-significant effects in the synthetic master set (p$<$.01). Upper and lower plots compare \texttt{Bayint} with methods that do not and do target selection, respectively. Right plots zoom in on main effects and most relevant interactions. \textbf{SBP} as outcome.}\label{rmsesbpsplit_synth}
\end{figure}

\end{document}